\documentclass[a4paper,10pt]{article}

\usepackage[utf8]{inputenc}
\usepackage[T1]{fontenc}
\usepackage{url}
\usepackage{geometry}
\usepackage{authblk}
\usepackage[parfill]{parskip}
\usepackage[bookmarks=false,colorlinks]{hyperref}
\usepackage[style=ieee,mincitenames=1,maxcitenames=1,maxbibnames=9,backend=biber]{biblatex}
\usepackage{amsthm}

\usepackage{tabularx}
\usepackage{enumitem}
\usepackage{mdl}
\usepackage{autonum}

\setlist{noitemsep,leftmargin=*}

\title{Minimax Optimal Additive Functional Estimation with Discrete Distribution: Slow Divergence Speed Case}
\author[1]{Kazuto Fukuchi~\thanks{kazuto@mdl.cs.tsukuba.ac.jp}}
\author[1,2,3]{Jun Sakuma~\thanks{jun@cs.tsukuba.ac.jp}}
\affil[1]{Department of Computer Science, Graduate School of System and Information Engineering, University of Tsukuba}
\affil[3]{JST CREST}
\affil[4]{RIKEN Center for Advanced Intelligence Project}

\begin{document}
\maketitle
\begin{abstract}
\noindent This paper addresses an estimation problem of an {\em additive functional} of $\phi$, which is defined as $\theta(P;\phi)=\sum_{i=1}^k\phi(p_i)$, given $n$ i.i.d. random samples drawn from a discrete distribution $P=(p_1,...,p_k)$ with alphabet size $k$. We have revealed in the previous paper~\autocite{DBLP:journals/corr/FukuchiS17} that the minimax optimal rate of this problem is characterized by the {\em divergence speed} of the fourth derivative of $\phi$ in a range of fast divergence speed. In this paper, we prove this fact for a more general range of the divergence speed. As a result, we show the minimax optimal rate of the additive functional estimation for each range of the parameter $\alpha$ of the divergence speed. For $\alpha \in (1,3/2)$, we show that the minimax rate is $\frac{1}{n}+\frac{k^2}{(n\ln n)^{2\alpha}}$. Besides, we show that the minimax rate is $\frac{1}{n}$ for $\alpha \in [3/2,2]$.
\end{abstract}

\section{Introduction}

Let $P$ be a probability measure with alphabet size $k$, where we use a vector representation of $P$; $P=(p_1,...,p_k)$ for $p_i=P\cbrace{i}$. Let $\phi$ be a mapping from $[0,1]$ to $\RealSet$. Given a set of i.i.d. samples $S_n=\cbrace{X_1,...,X_n} \sim P^n$, we deal with the problem of estimating an {\em additive functional} of $\phi$. The additive functional $\theta$ of $\phi$ is defined as
\begin{align}
 \theta(P;\phi) = \sum_{i=1}^k \phi(p_i).
\end{align}
We simplify this notation to $\theta(P;\phi)=\theta(P)$. Most of entropy-like criteria can be formed in terms of $\theta$. For instance, when $\phi(p)=-p\ln p$, $\theta$ is Shannon entropy. For a positive real $\alpha$, letting $\phi(p)=p^\alpha$, $\ln(\theta(P))/(1-\alpha)$ becomes R\'enyi entropy. More generally, letting $\phi = f$ where $f$ is a concave function, $\theta$ becomes $f$-entropies~\autocite{akaike1998information}. The estimation problem of such entropy-like criteria is a basic but important component for various research areas; such as physics~\autocite{lake2011accurate}, neuroscience~\autocite{nemenman2004entropy}, security~\autocite{gu2005detecting}, and machine learning~\autocite{quinlan1986induction,peng2005feature}.

The goal of this study is to derive the minimax optimal estimator of $\theta$ given a function $\phi$. To precisely define the minimax optimality, we introduce the (quadratic) minimax risk. A sufficient statistic of $P$ is a histogram $N=\paren{N_1,...,N_k}$, where letting $\ind{\cbrace*{\cdot}}$ be the indicator function, $N_j = \sum_{i=1}^n \ind{\cbrace{X_i = j}}$ and $N \sim \Mul(n, P)$. The estimator of $\theta$ can thus be defined as a function $\hat\theta:[n]^k\to\RealSet$, where $[m]=\cbrace{1,...,m}$ for an integer $m$. The quadratic minimax risk is defined as
\begin{align}
 R^*(n,k;\phi) = \inf_{\hat\theta}\sup_{P \in \dom{M}_k} \Mean\bracket*{\paren*{\hat\theta(N) - \theta(P)}^2}, \label{eq:minimax-risk}
\end{align}
where $\dom{M}_k$ is the set of all probability measures on $[k]$, and the infimum is taken over all estimators $\hat\theta$. By definition, we can interpret the minimax risk as the worst case risk of the best estimator. With this definition, an estimator $\hat\theta$ is minimax \mbox{(rate-)optimal} if there is a constant $C>0$ such that
\begin{align}
 \sup_{P \in \dom{M}_k} \Mean\bracket*{\paren*{\hat\theta(N) - \theta(P)}^2} \le C R^*(n,k;\phi).
\end{align}
Since no estimator achieves smaller worst case risk than the minimax risk, we can say that the minimax optimal estimator is best regarding the worst case risk.

Minimax optimality of the additive functional estimation problem has attracted much attention of many researchers. For fixed $k$, asymptotic efficiency and minimax optimality were proved if we employ the plugin or the maximum likelihood estimator, in which the estimated value is obtained by substituting the empirical mean of the probabilities $P$ into $\theta$~\autocite{van2000asymptotic}. However, the plugin estimator suffers from a large bias if the alphabet size $k$ is large against the sample size $n$. Indeed, the plugin estimators for $\phi(p)=-p\ln p$ and $\phi(p)=p^\alpha$ have been shown to be suboptimal in the large-$k$ regime in recent studies~\autocite{wu2016minimax,jiao2015minimax,DBLP:conf/soda/AcharyaOST15}. Recent studies investigated the minimax optimal estimators for $\phi(p)=-p\ln p$, $\phi(p)=p^\alpha$, and $\phi(p)=\ind{p > 0}$ in the large-$k$ regime~\autocite{DBLP:conf/soda/AcharyaOST15,jiao2015minimax,wu2016minimax,2015arXiv150401227W,2015arXiv150401227W}. However, the results of these studies were specifically derived for these $\phi$.

We previously derived the minimax optimal estimator for general $\phi \in C^4[0,1]$ such that $\abs*{\phi^{(4)}(p)}$ is finite for $p \in (0,1]$, where $C^m[0,1]$ denotes a class of $m$th times differentiable functions from $[0,1]$ to $\RealSet$ and $\phi^{(m)}$ is $m$th derivative of $\phi$~\autocite{DBLP:journals/corr/FukuchiS17}. We revealed that the minimax optimal risk can be characterized by the {\em divergence speed} of the fourth derivative of $\phi$, or referring it to the {\em forth divergence speed} of $\phi$ in this manuscript. The definition of the $\ell$th divergence speed is given as follows.
\begin{definition}[Divergence speed]\label{def:div-speed}
  For a positive integer $\ell$ and $\alpha \in \RealSet$, the $\ell$th divergence speed of $\phi \in C^\ell[0,1]$ is $p^{\alpha}$ if there exist constant $W_\ell > 0$, $c_\ell \ge 0$ and $c'_\ell \ge 0$ such that for all $p \in (0,1)$,
  \begin{align}
    W_\ell p^{-\ell+\alpha} - c'_\ell \le \abs*{\phi^{(\ell)}(p)} \le W_\ell p^{-\ell+\alpha} + c_\ell.
  \end{align}
\end{definition}
With this definition, we derived the minimax optimal estimator for the additive functional of $\phi$ of which the second divergence speed is $p^\alpha$ for $\alpha \in (0,1)$. The optimal estimator is constructed from combination of {\em best polynomial estimator} and {\em (second order) bias corrected plugin estimator}. The best polynomial estimator is an unbiased estimator of a polynomial that minimizes the $L_\infty$ approximation error against $\phi$. The (second order) bias corrected estimator is the plugin estimator with Miller's bias correction~\autocite{miller1955nbi}. While the previous characterization is valid for a range of fast divergence speed as $\alpha \in (0,1)$, the slower case as $\alpha \ge 1$ still remains as an open problem.

\subsection{Related Work}\label{sec:related-work}
Mane researchers have been dealing with the estimation problem of the additive functional and provides many estimators and analyses in decades past. The plugin estimator or the maximum likelihood estimator~(MLE) is the simplest way to estimate $\theta$, in which the empirical probabilities $\tilde{P} = (N_1/n,...,N_k/n)$ are substituted as $\theta(\tilde{P})$. The plugin estimator is asymptotically efficient and minimax optimal if $k$ is fixed~\autocite{van2000asymptotic}, whereas it is inconsistent if $k$ is larger than linear order of $n$. Bias-correction methods, such as \autocite{miller1955nbi,grassberger1988finite,10.2307/1936227}, can be applied to the plugin estimator to reduce the bias whereas these bias-corrected estimators are still suboptimal in the large-$k$ regime. The estimators based on Bayesian approaches in \autocite{schurmann1996entropy,6620615,holste1998bayes} are also suboptimal~\autocite{DBLP:journals/corr/HanJW15a}. Hence, the recent studies interested in analyzing the minimax optimal risk for the large-$k$ regime.

\textcite{paninski2004estimating} firstly revealed existence of a consistent estimator even if the alphabet size $k$ is larger than linear order of the sample size $n$. However, they did not provide a concrete form of the consistent estimator. The first estimator that achieves consistency in the large-$k$ regime is proposed by \textcite{DBLP:conf/stoc/ValiantV11}. However, minimax optimality of the estimator was not shown even in a more detailed analysis in \autocite{DBLP:conf/focs/ValiantV11}.

Recently, many researchers were interested in deriving the minimax optimal rate for the additive functionals in the large-$k$ regime. \textcite{DBLP:conf/soda/AcharyaOST15} showed that the bias-corrected estimator of R\'enyi entropy achieves the minimax optimal rate in regard to the sample complexity if $\alpha > 1$ and $\alpha \in \NaturalSet$, but they did not show the minimax optimality for other $\alpha$. \textcite{2015arXiv150401227W} derived a minimax optimal estimator for $\phi(p)=\ind{p > 0}$. For $\phi(p)=-p\ln p$, \textcite{jiao2015minimax,wu2016minimax} independently introduced the minimax optimal estimators in the large-$k$ regime. \textcite{7997814,jiao2015minimax} introduced a minimax optimal estimator for $\phi(p)=p^\alpha$ for any $\alpha > 0$ in the large-$k$ regime. However, their analysis is specifically designed for these $\phi$, and thus these results cannot be applied to other $\phi$.

We previously analyze the minimax optimal rate for the additive functional with characterization of the {\em divergence speed} in \textcite{DBLP:journals/corr/FukuchiS17}. Under the condition that the divergence speed of $\phi$ is $p^\alpha$, we revealed the optimal rate for $\alpha \in (0,1)$. For $\alpha \in (0,1/2)$, the minimax optimal rate is
\begin{align}
  \frac{k^2}{(n\ln n)^{2\alpha}},
\end{align}
as long as $k \gtrsim \ln^{4/3}n$. For $\alpha \in [1/2,1)$, the optimal rate was obtained as
\begin{align}
  \frac{k^2}{(n\ln n)^{2\alpha}} + \frac{k^{2-2\alpha}}{n}.
\end{align}
However, the minimax optimal rate for $\alpha \ge 1$ still remains as an open problem.

Although it is a special case, \textcite{jiao2015minimax,wu2016minimax} revealed the minimax optimal rate for $\alpha \ge 1$; the divergence speed of $\phi(p)=-p\ln p$ and $\phi(p)=p^\alpha$ are $p^1$ and $p^\alpha$, respectively. For $\phi(p) = p^\alpha$ where $p \in (1,3/2)$, \textcite{jiao2015minimax} showed the minimax optimal rate as
\begin{align}
  \frac{1}{(n\ln n)^{2\alpha-2}},
\end{align}
whereas it is valid only if $k \asymp n\ln n$. The minimax optimal rate for $\alpha \in (1,3/2)$ when the condition $k \asymp n\ln n$ is not satisfied is an interested research area for clear understanding in estimating the additive functional.

Besides, the optimal estimators for divergences with a large alphabet size have been investigated in \autocite{7541473,7840425,7541399,8006529}. The estimation problems of divergences are much complicated than the additive function, while the similar techniques were applied to derive the minimax optimality.

\subsection{Our Contribution}

\begin{table}[tb]
  \centering
  \caption{Summary of results. }\label{tbl:results-summary}
  \begin{tabularx}{\textwidth}{ll|lp{.4\textwidth}l}
    \hline
     $\alpha$ & $\ell$ & minimax rate & estimator & \\
    \hline
     $\le 0$ & $1$ & no consistent estimator & & \autocite{DBLP:journals/corr/FukuchiS17} \\
     $(0,1/2)$ & $4$ & $\frac{k^2}{(n\ln n)^{2\alpha}}$ if $k \gtrsim \ln^{\frac{4}{3}} n$ & best polynomial \& second-order bias corrected plugin & \autocite{DBLP:journals/corr/FukuchiS17} \\
     $[1/2,1)$ & $4$ & $\frac{k^2}{(n\ln n)^{2\alpha}}+\frac{k^{2-2\alpha}}{n}$ & best polynomial \& second-order bias corrected plugin & \autocite{DBLP:journals/corr/FukuchiS17} \\
     $(1,3/2)$ & $6$ & $\frac{k^2}{(n\ln n)^{2\alpha}}+\frac{1}{n}$ & best polynomial \& fourth-order bias corrected plugin & this paper \\
     $[3/2,2]$ & $2$ & $\frac{1}{n}$ & plugin & this paper \\
    \hline
  \end{tabularx}
\end{table}

In this paper, we derive the minimax optimal rate of the additive functional estimation and construct minimax optimal estimators for any $\alpha > 1$. The results are summarized in \cref{tbl:results-summary}. This table shows the minimax optimal rate~(third column) and the estimator that achieves the optimal rate~(fourth column) for each range of $\alpha$. The column $\ell$~(second column) means that the presented minimax optimality is valid if the $\ell$th divergence speed of $\phi$ is $p^\alpha$. Our contribution consists of revealing the minimax optimal rates and constructing the minimax optimal estimators for two range of $\alpha$; $\alpha \in (1,3/2)$ and $\alpha \in [3/2,2]$.


First, we show the minimax optimal rate and an optimal estimator for $\alpha \in (1,3/2)$. The existing result for this range is the analysis given by \textcite{jiao2015minimax} for the case $\phi(p) = p^\alpha$. As shown in \cref{sec:related-work}, they showed that the minimax optimal rate is $1/(n\ln n)^{2\alpha-2}$ under the condition $k \asymp n\ln n$. The optimal rate requires a strong condition on the relationship between $k$ and $n$, and the minimax optimal rate with $\alpha \in (1,3/2)$ is, therefore, far from clear understanding. In contrast, we success to prove the following minimax optimal rate without condition on the relationship between $k$ and $n$ as
\begin{align}
  \frac{k^2}{(n\ln n)^{2\alpha}} + \frac{1}{n}. \label{eq:optimal-rate-1-3/2-cont}
\end{align}
With this analysis, we clarify the question how many number of samples are necessary to consistently estimate the additive functional. To prove the upper bound in \cref{eq:optimal-rate-1-3/2-cont}, we employ the fourth order bias correction; which is an extension of the technique of Miller's bias correction~\autocite{miller1955nbi}. This technique offsets the fourth order approximation of bias of the plugin estimator in the similar manner of Miller's bias correction which offsets the second order approximation of bias. To prove the lower bound in \cref{eq:optimal-rate-1-3/2-cont}, we extend the technique introduced by \textcite{wu2016minimax} which analyzes the minimax lower bound for Shannon entropy. In the technique, the minimax lower bound is connected to the error of polynomial approximation. We carefully analyze the lower bound on the polynomial approximation error, and it yields the lower bound in \cref{eq:optimal-rate-1-3/2-cont}~(\cref{tbl:results-summary}, fourth row).

Second, we reveal that the plugin estimator is minimax optimal for $\alpha \in [3/2,2]$. In this case, the minimax optimal rate is $1/n$. It is almost obvious if $\alpha = 2$ because the second derivative is bounded in this case. To deal with the non-trivial case $\alpha \in [3/2,2)$, we extend the analysis given by \textcite{7997814}. They connected the bias of the plugin estimator to the Bernstein polynomial approximation error, and utilize the moduli of smoothness of $\phi$ to bound the the Bernstein polynomial approximation error. To apply their result, we analyze the moduli of smoothness of $\phi$ by using the divergence speed assumption. Besides, we extend their analysis to be applicable for general $\phi$~(\cref{tbl:results-summary}, fifth row)

\noindent{\bfseries Notations.}
We now introduce some additional notations. For any positive real sequences $\cbrace{a_n}$ and $\cbrace{b_n}$, $a_n \gtrsim b_n$ denotes that there exists a positive constant $c$ such that $a_n \ge c b_n$. Similarly, $a_n \lesssim b_n$ denotes that there exists a positive constant $c$ such that $a_n \le c b_n$. Furthermore, $a_n \asymp b_n$ implies $a_n \gtrsim b_n$ and $a_n \lesssim b_n$. For an event $\event$, we denote its complement by $\event^c$. For two real numbers $a$ and $b$, $a \lor b = \max\cbrace{a, b}$ and $a \land b = \min\cbrace{a, b}$.

\section{Preliminaries}

\subsection{Poisson Sampling}
We employ the Poisson sampling technique to derive upper and lower bounds for the minimax risk. The Poisson sampling technique models the samples as independent Poisson distributions, while the original samples follow a multinomial distribution. Specifically, the sufficient statistic for $P$ in the Poisson sampling is a histogram $\tilde{N} = \paren{\tilde{N}_i,...,\tilde{N}_k}$, where $\tilde{N}_1,...,\tilde{N}_k$ are independent random variables such that $\tilde{N}_i \sim \Poi(np_i)$. The minimax risk for Poisson sampling is defined as follows:
\begin{align}
 \tilde{R}^*(n,k;\phi) = \inf_{\cbrace{\hat\theta}}\sup_{P \in \dom{M}_k} \Mean\bracket*{\paren*{\hat\theta(\tilde{N}) - \theta(P)}^2}.
\end{align}
The minimax risk of Poisson sampling well approximates that of the multinomial distribution as $R^*(n,k;\phi) \asymp \tilde{R}^*(n,k;\phi)$ in ~\autocite{jiao2015minimax}.

\subsection{Best Polynomial Approximation}
\textcite{cai2011testing} presented a technique of the best polynomial approximation for deriving the minimax optimal estimators and their lower bounds for the risk. Let $\dom{P}_L$ be the set of polynomials of degree $L$. Given a function $\phi$ defined on an interval $I \subseteq [0,1]$ and a polynomial $g$, the $L_\infty$ error between $\phi$ and $g$ is defined as $\sup_{x \in I}\abs*{\phi(x) - p(x)}$. The best polynomial of $\phi$ with a degree-$L$ polynomial is a polynomial $g \in \dom{P}_L$ that minimizes the $L_\infty$ error. The error of the best polynomial approximation is defined as $E_L\paren*{\phi, I} = \inf_{g \in \dom{P}_L}\sup_{x \in I}\abs*{\phi(x) - g(x)}$. The error rate with respect to the degree $L$ has been studied since the 1960s~\autocite{timan1965theory,petrushev2011rational,ditzian2012moduli,achieser2013theory}. The polynomial that achieves the best polynomial approximation can be obtained, for instance, by the Remez algorithm~\autocite{remez1934determination} if $I$ is bounded.

\subsection{Basic Estimator Construction for $\alpha \in (0,3/2)$}\label{sec:basic-construction}
In the line of literature~\autocite{DBLP:conf/soda/AcharyaOST15,jiao2015minimax,wu2016minimax,DBLP:journals/corr/FukuchiS17}, the basic construction of the optimal estimator for the additive functional has a common part. Here, we describe the common methodology to construct the optimal estimator. For simplicity, we assume that we first draw $n' \sim \Poi(2n)$, and then draw $n'$ i.i.d. samples $S_{n'} = \cbrace{X_1,...,X_{n'}}$. Given the samples $S_{n'}$, we randomly divide these samples into two chunks. The assigned chunk for each sample is determined by a uniform random variable $B_i$, namely, $\p\cbrace{B_i = 0} = \p\cbrace{B_i = 1} = 1/2$. Then, we construct two hectograms $\tilde{N}$ and $\tilde{N}'$ defined as
\begin{align}
  \tilde{N}_i = \sum_{j=1}^{n'}\ind{B_j = 0}\ind{X_j = i}, \tilde{N}'_i = \sum_{j=1}^{n'}\ind{B_j = 1}\ind{X_j = i}.
\end{align}
$\tilde{N}$ and $\tilde{N}'$ are independent histograms such that $\tilde{N}_i,\tilde{N}'_i \sim \Poi(np_i)$.

Given $\tilde{N}'$, we determine whether the bias-corrected plugin estimator or the best polynomial estimator should be employed for each alphabet. Let $\Delta_{n,k}$ be a threshold depending on $n$ and $k$ to determine which estimator is employed. We apply the best polynomial estimator if $\tilde{N}'_i < 2\Delta_{n,k}$, and otherwise, i.e., $\tilde{N}'_i \ge 2\Delta_{n,k}$, we apply the bias-corrected plugin estimator. Let $\phi_{\mathrm{poly}}$ and $\phi_{\mathrm{plugin}}$ be the best polynomial estimator and the bias-corrected plugin estimator for $\phi$, respectively. Then, the estimator of $\theta$ is written as
\begin{align}
  \hat\theta(\tilde{N})\!=\!\sum_{i=1}^k \paren*{ \ind{\tilde{N}'_i \ge 2\Delta_{n,k}}\phi_{\mathrm{plugin}}(\tilde{N}_i)\!+\!\ind{\tilde{N}'_i < 2\Delta_{n,k}}\phi_{\mathrm{poly}}(\tilde{N}_i) }. \label{eq:hat-theta}
\end{align}
The best polynomial estimator $\phi_{\mathrm{poly}}$ is an unbiased estimator of the best polynomial of $\phi$ for a appropriate domain $I$. The detail of the best polynomial estimator can be found in \cref{sec:best-poly}. For $\phi_{\mathrm{plugin}}$, we previously employ the plugin estimator with Miller's bias correction~\autocite{DBLP:journals/corr/FukuchiS17}, which is described in \cref{sec:second-bias-correct}.

\subsection{Best Polynomial Estimator}\label{sec:best-poly}
Let $\cbrace{a_m}_{m=0}^L$ be coefficients of the polynomial that achieves the best approximation of $\phi$ by a degree-$L$ polynomial with range $I=[0,\frac{4\Delta_{n,k}}{n}]$; that is, the best approximation polynomial of $\phi$ is written as
\begin{align}
 \phi_L(p_i) =& \sum_{m=0}^L a_m p_i^m. \label{eq:approx-poly}
\end{align}
We utilize the factorial moments to construct an unbiased estimator of \cref{eq:approx-poly}. The $m$th factorial moment is defined as $(\tilde{N}_i)_m = \tfrac{\tilde{N}_i!}{(\tilde{N}_i - m)!}$ of which expectation is $(np_i)^m$. Thus, the following estimator is an unbiased estimator of $\phi_L$.
\begin{align}
  \bar\phi_{\mathrm{poly}}(\tilde{N}_i) = \sum_{m=0}^L \frac{a_m}{n^j} (\tilde{N}_i)_m.
\end{align}
We truncate $\bar\phi_{\mathrm{poly}}$ so that it is not outside of the domain of $\phi(p)$. Let $\phi_{{\mathrm{inf}},\frac{\Delta_{n,k}}{n}} = \inf_{p \in [0,\frac{\Delta_{n,k}}{n}]} \phi(p)$ and $\phi_{{\mathrm{sup}},\frac{\Delta_{n,k}}{n}} = \sup_{p \in [0,\frac{\Delta_{n,k}}{n}]} \phi(p)$. Then, the best polynomial estimator is defined as
\begin{align}
 \phi_{\mathrm{poly}}(\tilde{N}_i) = \paren*{\bar\phi_{\mathrm{poly}}(\tilde{N}_i) \land \phi_{{\mathrm{sup}},\frac{\Delta_{n,k}}{n}}} \lor \phi_{{\mathrm{inf}},\frac{\Delta_{n,k}}{n}}.
\end{align}

\subsection{Second Order Bias Correction Estimator}\label{sec:second-bias-correct}
In \autocite{DBLP:journals/corr/FukuchiS17}, we employ the plugin estimator with Miller's bias correction~\autocite{miller1955nbi} as $\phi_{\mathrm{plugin}}$. The bias correction offsets the second order bias, which is obtained as follows.
\begin{align}
  \Mean\bracket*{\phi\paren*{\frac{\tilde{N}}{n}} - \phi(p)} \approx& \Mean\bracket*{\frac{\phi^{(2)}(p)}{2}\paren*{\frac{\tilde{N}}{n} - p}^2} = \frac{p\phi^{(2)}(p)}{2n},
\end{align}
where $\tilde{N} \sim \Poi(np)$ for $p \in (0,1)$. The bias corrected function is hence obtained as $\phi_2(p) = \phi(p) - \frac{p\phi^{(2)}(p)}{2n}$.

Besides, we converted $\phi_2$ so that it becomes smooth by using the generalized Hermite interpolation~\autocite{spitball1960generalization}. The generalized Hermite interpolation between $\phi(a)$ and $\phi(b)$ is obtained as
\begin{align}
  H_L(p;\phi,a,b) = \phi(a) + \sum_{m=1}^L\frac{\phi^{(m)}(a)}{m!}(p-a)^m\sum_{\ell=0}^{L-m}\frac{L+1}{L+\ell+1}\Beta_{\ell,L+\ell+1}\paren*{\frac{p-a}{b-a}},
\end{align}
where $\Beta_{\nu,n}(x)=\binom n\nu x^\nu(1-x)^{n-\nu}$ denotes the Bernstein basis polynomial. Then, $H_L^{(i)}(a;\phi,a,b) = \phi^{(i)}(a)$ for $i = 0,...,L$ and $H_L^{(i)}(b;\phi,a,b) = 0$ for $i = 1,...,L$. Given an integer $L > 0$, an positive real $\Delta > 0$, and a function $\phi$, define a functional:
\begin{align}
  H_{L,\Delta}[\phi](p) = \begin{dcases}
    H_L\paren*{\frac{\Delta}{2};\phi,\Delta,\frac{\Delta}{2}} & \textif p \le \frac{\Delta}{2}, \\
    H_L\paren*{p;\phi,\Delta,\frac{\Delta}{2}} & \textif \frac{\Delta}{2} < p < \Delta, \\
    \phi(p) & \textif \frac{\Delta_{n,k}}{n} \le p \le 2, \\
    H_L\paren*{p;\phi,1,2} & \textif 1 < p < 2, \\
    H_L\paren*{2;\phi,1,2} & \textif p \ge 2.
  \end{dcases}
\end{align}
For a function $\phi$ of which the $\ell$th divergence speed is $p^\alpha$, this functional satisfies $\abs*{H_{L,\Delta}[\phi](p)} \lesssim \Delta^{\alpha-s}$ for $s = \ceil{\alpha},...,L$, $\ell \ge L$, and any $p \in (0,1)$.

Combining the bias correction and the smoothing procedure yields the second order bias corrected plugin estimator. To ensure smoothness of $\phi_2$, we set $L = 4$. Then, the estimator is given as follows.
\begin{align}
  \phi_{\mathrm{plugin}}(\tilde{N}_i) = H_{4,\frac{\Delta_{n,k}}{n}}[\phi]\paren*{\frac{\tilde{N}_i}{n}} - \frac{\tilde{N}_i}{2n^2}H^{(2)}_{4,\frac{\Delta_{n,k}}{n}}\bracket*{\phi}\paren*{\frac{\tilde{N}_i}{n}}.
\end{align}

\section{Main Results}

Our main results are revealing the minimax optimal rate of the additive functional estimation in characterizing with the divergence speed. We derive the minimax optimal rates for each range of $\phi$; $\alpha \in (1,3/2)$ and $\alpha \in [3/2,2]$. First, we derive the minimax optimal rate for $\alpha \in (1,3/2)$. In this case, the rate is obtained as follows.
\begin{theorem}\label{thm:optimal-rate-1-3/2}
  Suppose $\phi:[0,1]\to\RealSet$ is a function of which the sixth divergence speed is $p^{\alpha}$ for $\alpha \in (1,3/2)$. If $n \gtrsim k^{1/\alpha}/\ln k$, the minimax optimal rate is obtained as
  \begin{align}
    R^*(n,k;\phi) \asymp \frac{k^2}{(n\ln n)^{2\alpha}} + \frac{1}{n}, \label{eq:optiaml-rate-1-3/2}
  \end{align}
  otherwise there is no consistent estimator.
\end{theorem}
For this range, \textcite{jiao2015minimax} showed the minimax optimal rate for $\phi(p) = p^{\alpha}$ as $\frac{1}{(n\ln n)^{2\alpha-2}}$ under condition that $k \asymp n\ln n$. In contrast, we success to prove the minimax optimal rate for this range without the condition $k \asymp n\ln n$. The first term in \cref{eq:optiaml-rate-1-3/2} corresponds to their result because it is same as their result under the condition they assume. To derive the upper bound, we introduce the fourth order bias correction which offsets bias of the plugin estimator up to fourth order Taylor approximation in the similar manner of the Miller's bias correction~\autocite{miller1955nbi}. This technique will be explained in \cref{sec:forth-order}. For proving the lower bound, we follows the same manner of the analysis given by \textcite{wu2016minimax}, in which the minimax lower bound is connected to the lower bound on the best polynomial approximation. Our careful analysis of the best polynomial approximation yields the lower bound~(in \cref{sec:lower-1-3/2}).

Second, the minimax optimal rate for $\alpha \in [3/2,2]$ is obtained as follows.
\begin{theorem}\label{thm:optimal-rate-3/2}
  Suppose $\phi:[0,1]\to\RealSet$ is a function of which second divergence speed is $p^{\alpha}$ for $\alpha \in [3/2,2]$. Then, the minimax optimal rate is obtained as
  \begin{align}
    R^*(n,k;\phi) \asymp \frac{1}{n}.
  \end{align}
\end{theorem}
\begin{remark}
  The meaning of the second divergence speed is $p^2$ is that the second derivative of $\phi$ is bounded. Since a function of which the $\ell$th divergence speed is $p^\alpha$ for any $\alpha \ge 2$ and $\ell \ge \alpha$ has bounded second derivative, this result covers all case for $\alpha \ge 3/2$.
\end{remark}
The lower bound can be obtained easily by application of LeCam's two point method~\autocite{LeCam:1986:AMS:20451}. The upper bound is obtained by employing the plugin estimator; its analysis is easy if $\alpha = 2$ because the second derivative of $\phi$ is bounded. For $\alpha \in [3/2,2)$, we extend the analysis of $\phi(p)=p^\alpha$ given by \textcite{7997814} to be applicable to general $\phi$ in \cref{sec:plugin-3/2-2}.

{\bfseries H\"older continuousness under the divergence speed assumption.}
Before moving to the detailed description of analyses, we show the useful property of $\phi$ regarding the H\"older continuousness, which comes from the divergence speed assumption. For a real $\bar\alpha \in (0,1]$, a function $\phi$ is $\bar\alpha$-H\"older continuous if
\begin{align}
  \norm{\phi}_{C^{0,\bar\alpha}} = \sup_{x \ne y \in I}\frac{\abs*{\phi(x) - \phi(y)}}{\abs*{x - y}^{\bar\alpha}} < \infty.
\end{align}
In particular, $1$-H\"older continuous is called as Lipschitz continuous. If the divergence speed of $\phi$ is $p^\alpha$ for $\alpha \in (1,2]$, the following properties regarding the H\"older continuousness hold.
\begin{lemma}\label{lem:div-speed-holder}
  Suppose $\phi:[0,1]\to\RealSet$ is a function of which second divergence speed is $p^{\alpha}$ for $\alpha \in (1,2]$ such that $\phi^{(1)}(0) = 0$. Then, $\phi$ is Lipschitz continuous, and $\phi^{(1)}$ is $(\alpha-1)$-H\"older continuous.
\end{lemma}
Note that we can assume $\phi^{(1)}(0)=0$ without loss of generality because, for any $c \in \RealSet$, $\theta(P;\phi) = \theta(P;\phi_c)$ where $\phi_c(p) = \phi(p) + c(p-1/k)$.

\section{Estimators and Upper Bound Analysis}

In this section, we introduce the minimax optimal estimators for each range of $\alpha$. For $\alpha \in (1,3/2)$, we employ the basic construction described in \cref{sec:basic-construction}. We thus describe construction of $\phi_{\mathrm{plugin}}$ in \cref{sec:forth-order}. Besides, we analyze the bias and the variance of $\phi_{\mathrm{plugin}}$ and $\phi_{\mathrm{poly}}$. For $\alpha \in [3/2,2]$, we employ the plugin estimator, whereas it is non-trivial from the existing results that the bias can be bounded by $1/n$. In \cref{sec:plugin-3/2-2}, we will prove the bound on the bias.

\subsection{Fourth order bias correction for $\alpha \in (1,3/2)$}\label{sec:forth-order}

As mentioned before, the second order bias correction offsets the second order approximation of bias. In analogy with that, the fourth order bias correction offsets the fourth order approximation of bias. By the Taylor approximation, the bias of $\phi_2$ is obtained as
\begin{align}
  \Mean\bracket*{\phi_2\paren*{\frac{\tilde{N}}{n}} - \phi(p)} \approx& \frac{2p\phi^{(3)}}{3n^2} + \frac{7p\phi^{(4)}}{24n^3} + \frac{3p^2\phi^{(4)}}{8n^2}.
\end{align}
Thus, the fourth order bias corrected function is obtained as $\phi_4(p) = \phi_2(p) - \frac{2p\phi^{(3)}}{3n^2} - \frac{7p\phi^{(4)}}{24n^3} - \frac{3p^2\phi^{(4)}}{8n^2}$.

As well as the second order bias correction, we employ the smoothing procedure using the generalized Hermite interpolation described in \cref{sec:second-bias-correct}. To ensure the smoothness of $\phi_4$, we use $L = 6$. Hence, the estimator is obtained as
\begin{multline}
  \phi_{\mathrm{plugin}}(\tilde{N}_i) = H_{6,\frac{\Delta_{n,k}}{n}}[\phi]\paren*{\frac{\tilde{N}_i}{n}} - \frac{\tilde{N}_i}{2n^2}H^{(2)}_{6,\frac{\Delta_{n,k}}{n}}\bracket*{\phi}\paren*{\frac{\tilde{N}_i}{n}} \\ - \frac{2\tilde{N}_i}{3n^3}H^{(3)}_{6,\frac{\Delta_{n,k}}{n}}\bracket*{\phi}\paren*{\frac{\tilde{N}_i}{n}} - \frac{7\tilde{N}_i}{24n^4}H^{(4)}_{6,\frac{\Delta_{n,k}}{n}}\bracket*{\phi}\paren*{\frac{\tilde{N}_i}{n}} - \frac{3\tilde{N}^2_i}{8n^4}H^{(4)}_{6,\frac{\Delta_{n,k}}{n}}\bracket*{\phi}\paren*{\frac{\tilde{N}_i}{n}}. \label{eq:fourth-bias-correct}
\end{multline}

For analysis of the estimator in \cref{eq:fourth-bias-correct}, we define a function as
\begin{align}
  \bar\phi_{4,\Delta}(p)\!=\!H_{6,\Delta}[\phi](p)\!-\!\frac{p}{2n}H^{(2)}_{6,\Delta}[\phi](p)\!-\!\frac{2p}{3n^2}H^{(3)}_{6,\Delta}[\phi](p)\!-\!\frac{7p}{24n^3}H^{(4)}_{6,\Delta}[\phi](p)\!-\!\frac{3p^2}{8n^2}H^{(4)}_{6,\Delta}[\phi](p).
\end{align}
We analyze the bias and the variance of $\bar\phi_{4,\Delta}(\tilde{N}/n)$.
\begin{lemma}\label{lem:plugin-bias-1-3/2}
 Suppose $\phi:[0,1]\to\RealSet$ is a function of which the sixth divergence speed is $p^{\alpha}$ for $\alpha \in (1,3/2)$. Suppose $\frac{1}{n} \lesssim \Delta < p \le 1$. Let $\tilde{N} \sim \Poi(np)$. Then, we have
 \begin{align}
  \Bias\bracket*{\bar\phi_{4,\Delta}\paren*{\frac{\tilde{N}}{n}} - \phi(p)} \lesssim \frac{1}{n^3\Delta^{3-\alpha}}.
 \end{align}
\end{lemma}
\begin{lemma}\label{lem:plugin-var-1-3/2}
 Suppose $\phi:[0,1]\to\RealSet$ is a function of which the fifth divergence speed is $p^{\alpha}$ for $\alpha \in (1,3/2)$. Suppose $\frac{1}{n} \lesssim \Delta < p \le 1$. Let $\tilde{N} \sim \Poi(np)$. Then, we have
 \begin{align}
   \Var\bracket*{\bar\phi_{4,\Delta}\paren*{\frac{\tilde{N}}{n}}} \lesssim \frac{1}{n^{2\alpha}} + \frac{p}{n}.
 \end{align}
\end{lemma}
We set the parameter $\Delta$ as $\Delta \asymp \frac{\ln n}{n}$. With this $\Delta$, we can see that the bias and the variance do not exceed the rate in \cref{thm:optimal-rate-1-3/2}.

\subsection{Best polynomial error analysis for $\alpha \in (1,3/2)$}
Here, we analyze the bias and the variance of the best polynomial estimator $\phi_{\mathrm{poly}}$. For the variance, we can use the following lemma shown in \autocite[Lemma 5]{DBLP:journals/corr/FukuchiS17}.
\begin{lemma}[{\autocite[Lemma 5]{DBLP:journals/corr/FukuchiS17}}]\label{lem:poly-var}
 Let $\tilde{N} \sim \Poi(np)$. Given an integer $L$ and a positive real $\Delta \gtrsim \frac{1}{n}$, let $\phi_L(p) = \sum_{m=0}^L a_mp^m$ be the optimal uniform approximation of $\phi$ by degree-$L$ polynomials on $[0,\Delta]$, and $g_L(\tilde{N}) = \sum_{m=0}^L a_m(\tilde{N})_m/n^m$ be an unbiased estimator of $\phi_L(p)$. Assume $\phi$ is bounded. If $p \le \Delta$ and $2\Delta^3L \le n$, we have
 \begin{align}
  \Var\bracket*{(g_L(\tilde{N}) \land \phi_{{\mathrm{sup}},\Delta})\lor \phi_{{\mathrm{inf}},\Delta} } \lesssim \frac{\Delta^3 L 64^L (2e)^{2\sqrt{\Delta n L }}}{n}.
 \end{align}
\end{lemma}

As well as the variance, we can obtain the same claim of \autocite[Lemma 4]{DBLP:journals/corr/FukuchiS17} as follows.
\begin{lemma}\label{lem:poly-bias}
 Let $\tilde{N} \sim \Poi(np)$. Given an integer $L$ and a positive real $\Delta$, let $\phi_L(p) = \sum_{m=0}^L a_mp^m$ be the optimal uniform approximation of $\phi$ by degree-$L$ polynomials on $[0,\Delta]$, and $g_L(\tilde{N}) = \sum_{m=0}^L a_m(\tilde{N})_m/n^m$ be an unbiased estimator of $\phi_L(p)$. If the fourth divergence speed of $\phi$ is $p^\alpha$ for $\alpha \in (1,3/2)$, we have
 \begin{align}
  \Bias\bracket*{(g_L(\tilde{N}) \land \phi_{{\mathrm{sup}},\Delta})\lor \phi_{{\mathrm{inf}},\Delta} - \phi(p) } \lesssim \sqrt{\Var\bracket*{g_L(\tilde{N})}} + \paren*{\frac{\Delta}{L^2}}^{\alpha}.
 \end{align}
\end{lemma}
However, we need to show the following lemma to prove \cref{lem:poly-bias}.
\begin{lemma}\label{lem:best-error}
  Suppose $\phi:[0,1]\to\RealSet$ is a function of which the third divergence speed is $p^{\alpha}$ for $\alpha \in (1,3/2)$. Then, we have
  \begin{align}
    E_L(\phi,[0,\Delta]) \lesssim \paren*{\frac{\Delta}{L^2}}^\alpha.
  \end{align}
\end{lemma}
To prove \cref{lem:best-error}, we use the Jackson’s inequality which gives a bound on the best trigonometric polynomial approximation error by using the first order moduli of smoothness, defined as
\begin{align}
  \omega_1(f,t) = \sup_{x, y \in (-\pi,\pi)}\cbrace{\abs*{f(x)-f(y)} : \abs*{x-y} \le t}
\end{align}
From the Jackson's inequality~\autocite{achieser2013theory}, any trigonometric polynomial $T_L$ with degree-$L$ satisfies
\begin{align}
  \sup_{x \in [0,2\pi]}\abs*{\phi(x) - T_L(x)} \lesssim \frac{1}{L^2}\omega_1(\phi^{(2)},L^{-1}).
\end{align}
\sloppy Letting $\phi_\Delta(x)=\phi(\Delta x^2)$, we have $E_L(\phi,[0,\Delta]) = E_L(\phi_\Delta,[-1,1])$. Using the fact that $\abs*{\cos^{-1}(x)-\cos^{-1}(y)} \ge \abs*{x-y}$ for $x,y \in (0,1)$, we have $E_L(\phi_\Delta,[-1,1]) = \sup_{x \in [0,2\pi]}\abs*{\phi(\cos(x)) - T_L(x)} \le \frac{1}{L^2}\omega_1(\phi^{(2)}_\Delta,L^{-1})$. We show the following bound regarding $\omega_1(\phi^{(2)}_\Delta,t)$.
\begin{lemma}\label{lem:first-moduli-bound}
  Suppose $\phi:[0,1]\to\RealSet$ is a function of which the third divergence speed is $p^{\alpha}$ for $\alpha \in (1,3/2)$. Then, we have
\begin{align}
  \omega_1(\phi^{(2)}_\Delta,t) \lesssim \Delta^{\alpha}L^{2\alpha-2}.
\end{align}
\end{lemma}
The proof of \cref{lem:first-moduli-bound} is left to the appendix. Consequently, the proof of \cref{lem:best-error} is obtained as follows.
\begin{proof}[Proof of \cref{lem:best-error}]
  Use \cref{lem:first-moduli-bound} and $E_L(\phi_\Delta,[-1,1]) \le \frac{1}{L^2}\omega_1(\phi^{(2)}_\Delta,L^{-1})$.
\end{proof}

\subsection{Analysis of plugin estimator with $\alpha \in [3/2,2)$}\label{sec:plugin-3/2-2}
For $\alpha \in [3/2,2]$, the plugin estimator is a minimax optimal estimator in which the optimal rate is $1/n$. To prove this, we analyze the bias and the variance of the plugin estimator. The variance is easily proved as follows.
\begin{theorem}\label{thm:variance-lipschitz}
  If $\phi$ is Lipschitz continuous,
  \begin{align}
    \Var\bracket*{\sum_{i=1}^n\phi\paren*{\frac{N_i}{n}}} \lesssim \frac{1}{n}.
  \end{align}
\end{theorem}
The proof of \cref{thm:variance-lipschitz} is accomplished by applying the concentration result of the bounded difference~(see e.g., \autocite{boucheron2013concentration}). The bias also is easily proved if $\alpha = 2$ because of the Lipschitz continuousness of $\phi^{(1)}$;
\begin{theorem}\label{thm:bias-lipschitz}
  If $\phi^{(1)}$ is Lipschitz continuous,
  \begin{align}
    \Bias\bracket*{\sum_{i=1}^n\phi\paren*{\frac{N_i}{n}} - \theta(P)} \lesssim \frac{1}{n}.
  \end{align}
\end{theorem}
This is simply obtained by using the Taylor theorem. In contrast, derivation of the bias bound for $\alpha \in (3/2,2)$ is not trivial.

To prove the bound on the bias for $\alpha \in (3/2,2)$, we extend the analysis given by \textcite{7997814}, in which they analyzed the plugin estimator for $\phi(p)=p^\alpha$. They showed that the bias of the plugin estimator is same as the polynomial approximation using the Bernstein polynomial. The error of the Bernstein polynomial approximation is bounded by the second order moduli of smoothness defined as
\begin{align}
  \omega^2(f,t) = \sup_{x,y \in [-1,1]}\cbrace*{\abs*{f(x)+f(y)-2f\paren*{\frac{x+y}{2}}} : \abs*{x-y} \le 2t}.
\end{align}
With this definition, the bias is bounded as follows.
\begin{lemma}[\textcite{7997814}]\label{lem:bias-moduli}
  Let $N \sim \Binomial(n,p)$. Given a function $\phi:[0,1]\to\RealSet$, we have
  \begin{align}
    \Bias\bracket*{\phi\paren*{\frac{N}{n}} - \phi(p)} \lesssim \omega^2\paren*{\phi,\sqrt{\frac{p(1-p)}{n}}}.
  \end{align}
\end{lemma}
To take advantage of \cref{lem:bias-moduli}, we derive the bound on the second order moduli of smoothness. If $\alpha \in [3/2,2)$, the bound on $\omega^2(\phi,t)$ is obtained as follows.
\begin{lemma}\label{lem:moduli-bound}
  Suppose $\phi:[0,1]\to\RealSet$ is a function of which second divergence speed is $p^{\alpha}$ for $\alpha \in [3/2,2)$, where $\phi^{(1)}(0) = 0$. Then, we have
  \begin{align}
    \omega^2(\phi,t) \lesssim t^{\alpha}.
  \end{align}
\end{lemma}
\begin{proof}
  As shown by \textcite{devore1993constructive}, $\omega^2(\phi,t) \lesssim t\omega^1(\phi^{(1)},t)$ where
  \begin{align}
    \omega^1(f,t) = \sup_{x,y \in [-1,1]}\cbrace*{\abs*{f(x)-f(y)} : \abs*{x-y} \le 2t}.
  \end{align}
  By definition, $\omega^1(\phi^{(1)},t) \lesssim t^{\alpha-1}$ because of the H\"older continuousness of $\phi^{(1)}$ from \cref{lem:div-speed-holder}.
\end{proof}
By utilizing \cref{lem:bias-moduli,lem:moduli-bound}, we prove the bias.
\begin{theorem}\label{thm:bias-plugin}
  Suppose $\phi:[0,1]\to\RealSet$ is a function of which second divergence speed is $p^{\alpha}$ for $\alpha \in [3/2,2)$ such that $\phi(0)=0$. Then, we have
  \begin{align}
    \Bias\bracket*{\sum_i\phi\paren*{\frac{N_i}{n}} - \theta(P)} \lesssim \frac{1}{n^{\alpha-1}}.
  \end{align}
\end{theorem}
With the bias-variance decomposition, the squared error is bounded by the sum of the variance and the squared bias. From \cref{thm:bias-plugin}, the squared bias is bounded as $\frac{1}{n^{2\alpha-2}}$; this is less than $\frac{1}{n}$ if $\alpha \in [3/2,2)$.

\section{Lower Bound Analysis}

We here describe lower bound analyses and prove the lower bound of \cref{thm:optimal-rate-1-3/2,thm:optimal-rate-3/2}. First, we prove the $\frac{1}{n}$ term, which is accomplished by applying the LeCam's two point method with appropriate construction of two probability vectors. The precise claim is given by the following theorem.
\begin{theorem}\label{thm:lower1}
  Suppose $\phi:[0,1]\to\RealSet$ is a function of which second divergence speed is $p^{\alpha}$ for $\alpha \in (1,2]$. For any $n \ge 1$ and $k \ge 3$, we have
  \begin{align}
    R^*(n,k;\phi) \gtrsim \frac{1}{n}.
  \end{align}
\end{theorem}
The proof can be found in \cref{Sec:proof-thm-lower1}.

The remaining term is $k^2/(n\ln n)^{2\alpha}$ for $\alpha \in (1,3/2)$. We prove this term in the next subsection. To prove this term, we connect the minimax lower bound to the best polynomial approximation error. Besides, we show the error rate of the best polynomial approximation error. Combining these results, we prove the $k^2/(n\ln n)^{2\alpha}$ term.

\subsection{Lower Bound Analysis for $\alpha \in (1,3/2)$}\label{sec:lower-1-3/2}

Here, we will prove the following lower bound.
\begin{theorem}\label{thm:lower2}
  Suppose $\phi:[0,1]\to\RealSet$ is a function of which second divergence speed is $p^{\alpha}$ for $\alpha \in (1,3/2)$. If $n \gtrsim k^{1/\alpha}/\ln k$,
  \begin{align}
    R^*(n,k;\phi) \gtrsim \frac{k^2}{(n\ln n)^{2\alpha}}.
  \end{align}
\end{theorem}
The proof of \cref{thm:lower2} basically follows the same manner of \textcite{wu2016minimax}. \textcite{wu2016minimax} characterized the lower bound on the minimax risk for Shannon entropy by the best polynomial approximation error.

We firstly connect the minimax lower bound to the best polynomial approximation error. The precise claim is shown in the following theorem.
\begin{theorem}\label{thm:poly-approx-lower}
  Suppose $\phi$ is Lipschitz continuous, and $\phi^{(1)}$ is $(\alpha-1)$-H\"older continuous such that $\phi^{(1)}(0)=0$. Let $\phi^\star(p)=\phi(p)/p$. For any given integer $L > 0$, $\lambda \le k$, and $\gamma \in (0,1)$ such that $\gamma \le \lambda/k$, with $d = 2k\gamma E_L(\phi^\star, [\gamma,\lambda/k])$ the following holds
  \begin{multline}
   \tilde{R}^*(n/2,k;\phi) \ge \frac{d^2}{32}\paren[\Bigg]{\frac{7}{8} - k\paren*{\frac{2en\lambda}{kL}}^L - \frac{32\norm*{\phi^{(1)}}_{C^{0,\alpha-1}}^2\lambda^{2\alpha}}{k^{2\alpha-1}d^2} \\
  - \frac{128\norm*{\phi}_{C^{0,1}}^2e^{-n/32}}{d^2} - \frac{512\norm*{\phi}_{C^{0,1}}^2\lambda^2}{kd^2}}.
  \end{multline}
  as long as $\lambda/\sqrt{k} < 1/12$.
\end{theorem}
By \cref{thm:poly-approx-lower}, we characterize the minimax lower bound by the best polynomial approximation error $E_L(\phi^\star, [\gamma,\lambda/k])$.

We next derive the lower bound on $E_L(\phi^\star, [\gamma,\lambda/k])$.
\begin{theorem}\label{thm:lower-poly-approx}
  For $\gamma \in (0,1)$, let $\cbrace*{\phi_\gamma}$ be a family of functions such that $\phi_\gamma(\gamma) = 0$, $\phi_\gamma^{(1)}(0)=0$, and the second order divergence speed of all elements is $p^\alpha$ for $\alpha \in (1,3/2)$. Denote $\phi_\gamma^\star(x)=\phi_\gamma(x)/x$. Then, there exists an universal constant $c > 0$ such that
  \begin{align}
    \limsup_{L \to \infty, \gamma \to 0 : \gamma \le 1/2L^2}\gamma^{1-\alpha}E_L\paren*{\phi_\gamma^\star, [\gamma,2L^2\gamma]} > c.
  \end{align}
\end{theorem}
\begin{remark}
  We can choose such family $\cbrace*{\phi_\gamma}$ because the minimax risk is invariant among $\phi_{c,c'}(x) = \phi(x)+c+c'(x-1/k)$ for any constants $c,c' \in \RealSet$.
\end{remark}
As proved by \cref{thm:lower-poly-approx}, we get $d \gtrsim k\gamma^\alpha$ as long as $\lambda/k = 2L^2\gamma$. Fix $\delta > 0$. Set $\gamma \asymp \frac{1}{n^{1+\delta}\ln n}$ and $L \asymp \ln n$. Then, we have from \cref{thm:poly-approx-lower,thm:lower-poly-approx} that for any sufficiently small $\delta > 0$,
\begin{align}
  \tilde{R}^*(n/2,k;\phi) \gtrsim \frac{k^2}{(n^{1+\delta}\ln n)^\alpha}.
\end{align}
By arbitrariness of $\delta > 0$, we get the lower bound on \cref{thm:optimal-rate-1-3/2}.

\section{Discussion}
In this paper, we reveal that the divergence speed characterizes the minimax optimal rate for $\alpha \in (1,3/2)$ and $\alpha \in [3/2,2]$. Combining the previous result in \autocite{DBLP:journals/corr/FukuchiS17}, the minimax rate is characterized by the divergence speed for any $\alpha \in (0,2]$ except $\alpha = 1$. The Shannon entropy case, i.e., $\phi(p)=-p\ln p$, is one special case in which the divergence speed of $\phi$ is $p^1$. The existing result for $\phi(p)=-p\ln p$ was given by \textcite{jiao2015minimax,wu2016minimax}, in which they proved the minimax optimal rate as $\frac{k^2}{(n\ln n)^2} + \frac{\ln^2 k}{n}$. We can expect that the same rate is obtained for $\alpha=1$. However, this still remains as an open problem.

\printbibliography

\appendix

\section{Proof of \cref{lem:div-speed-holder}}
\begin{proof}[Proof of \cref{lem:div-speed-holder}]
  The Lipschitz continuousness is proved by showing there exists an universal constant $C > 0$ such that
  \begin{align}
    \sup_{p \in (0,1)}\abs*{\phi^{(1)}(p)} \le C.
  \end{align}
  For any $p \in (0,1)$, the absolutely continuousness of $\phi^{(1)}$ gives
  \begin{align}
    \abs*{\phi^{(1)}(p)} =& \abs*{\int_0^p\phi^{(2)}(s)ds} \\
    \le& \int_0^p\abs*{\phi^{(2)}(s)}ds \\
    \le& \int_0^p\paren*{W_2s^{\alpha-2} + c_2}ds \\
    =& \frac{W_2}{\alpha-1}p^{\alpha-1} + c_2p \le \frac{W_2}{\alpha-1} + c_2.
  \end{align}

  Next, we prove the H\"older continuousness of $\phi^{(1)}$. The H\"older continuousness is proved by showing there exists an universal constant $C > 0$ such that for any $x,y \in (0,1)$,
  \begin{align}
    \abs*{\phi^{(1)}(x) - \phi^{(1)}(y)} \le C\abs*{x - y}^{\alpha-1}.
  \end{align}
  The absolutely continuousness of $\phi^{(1)}$ yields for any $x,y \in (0,1)$,
  \begin{align}
    \abs*{\phi^{(1)}(x) - \phi^{(1)}(y)} \le& \abs*{\int_x^y\phi^{(2)}(s)ds} \\
    \le& \int_x^y\abs*{\phi^{(2)}(s)}ds \\
    \le& \int_x^y\paren*{W_2s^{\alpha-2} + c_2}ds \\
    =& \frac{W_2}{\alpha-1}\paren*{x^{\alpha-1} - y^{\alpha-1}} + c_2\paren*{x-y}.
  \end{align}
  Since a function $x \to x^\beta$ for $\beta \in (0,1)$ is $\beta$-H\"older continuous and $\abs{x-y} \le 1$ for any $x,y \in (0,1)$, we have
  \begin{align}
    \abs*{\phi^{(1)}(x) - \phi^{(1)}(y)} \le& \frac{W_2}{\alpha-1}\abs*{x-y}^{\alpha-1} + c_1\abs*{x-y}^{\alpha-1}.
  \end{align}
\end{proof}

\section{Detailed Analysis of Fourth Order Bias Corrected Plugin Estimator}

We firstly prove that the smoothing procedure ensures $p^\beta\abs*{H^{(\ell)}_{L,\Delta}[\phi](p)} \lesssim \Delta^{\alpha+\beta-\ell}$.
\begin{lemma}\label{lem:hermite-bound}
  For $\alpha$, let $\phi:[0,1]\to\RealSet$ be a function of which $L$the divergence speed is $p^\alpha$, where $L > \alpha$ is an universal constant. For $\ell \le L$ and $\beta$ such that $\ell > \alpha + \beta$, we have
  \begin{align}
    \sup_{p > 0}p^\beta\abs*{H^{(\ell)}_{L,\Delta}[\phi](p)} \lesssim \Delta^{\alpha+\beta-\ell}.
  \end{align}
  Moreover, for $\ell \le L$ and $\beta$ such that $1 \le \ell \le \alpha + \beta$, we have
  \begin{align}
    \sup_{p > 0}p^\beta\abs*{H^{(\ell)}_{L,\Delta}[\phi](p)} \lesssim 1.
  \end{align}
\end{lemma}
Using \cref{lem:hermite-bound}, we prove the bias and the variance of the bias corrected plugin estimator.
\begin{proof}[Proof of \cref{lem:plugin-bias-1-3/2}]
  Application of the Taylor theorems yields that
  \begin{align}
    & \Bias\bracket*{\bar\phi_{4,\Delta}\paren*{\frac{\tilde{N}}{n}}-\phi(p)} \\
    =& \begin{multlined}[t]
     \abs[\Bigg]{
       \Mean\bracket*{\frac{p}{2n}\phi^{(2)}(p)-\frac{\tilde{N}}{2n^2}H^{(2)}_{6,\Delta}[\phi]\paren*{\frac{\tilde{N}}{n}}}
       +\Mean\bracket*{\frac{p}{6n^2}\phi^{(3)}(p)-\frac{2\tilde{N}}{3n^3}H^{(3)}_{6,\Delta}[\phi]\paren*{\frac{\tilde{N}}{n}}} \\
       +\Mean\bracket*{\frac{p}{24n^3}\phi^{(4)}(p)-\frac{7\tilde{N}}{24n^4}H^{(4)}_{6,\Delta}[\phi]\paren*{\frac{\tilde{N}}{n}}}
       +\Mean\bracket*{\frac{p^2}{8n^2}\phi^{(4)}(p)-\frac{3\tilde{N}^2}{8n^4}H^{(4)}_{6,\Delta}[\phi]\paren*{\frac{\tilde{N}}{n}}} \\
       + \frac{p^2}{12n^3}\phi^{(5)}(p) + \frac{p}{120n^4}\phi^{(5)}(p) + \Mean\bracket*{R_5\paren*{\frac{\tilde{N}}{n};H_{6,\Delta},p}}
     },
    \end{multlined}
  \end{align}
  where $R_5$ denotes the reminder term of the Taylor theorem. Besides, we have
  \begin{align}
    &\Mean\bracket*{\frac{p}{2n}\phi^{(2)}(p)-\frac{\tilde{N}}{2n^2}H^{(2)}_{6,\Delta}[\phi]\paren*{\frac{\tilde{N}}{n}}} \\
    =& \Mean\bracket*{\frac{1}{2n}\phi^{(2)}(p)\paren*{p-\frac{\tilde{N}}{n}}+\frac{\tilde{N}}{2n}\paren*{\phi^{(2)}(p)-H^{(2)}_{6,\Delta}[\phi]\paren*{\frac{\tilde{N}}{n}}}} \\
    =& \begin{multlined}[t]
      \frac{p}{2n^2}\phi^{(3)}(p) + \frac{p}{4n^3}\phi^{(4)}(p) + \frac{p^2}{4n^2}\phi^{(4)}(p) \\ + \frac{p}{12n^4}\phi^{(5)}(p) + \frac{p^2}{3n^3}\phi^{(5)}(p) + \Mean\bracket*{\frac{\tilde{N}}{2n^2}R_3\paren*{\frac{\tilde{N}}{n};H^{(2)}_{6,\Delta},p}},
    \end{multlined}
  \end{align}
  where $R_4$ is the reminder term of the Taylor theorem. We use $\Mean[X(X-\lambda)]=\lambda$, $\Mean[X(X-\lambda)^2]=\lambda^2+\lambda$, $\Mean[X(X-\lambda)^3]=4\lambda^2+\lambda$ for $X \sim \Poi(\lambda)$ to get the last line. Hence,
  \begin{align}
    & \Bias\bracket*{\bar\phi_{4,\Delta}\paren*{\frac{\tilde{N}}{n}}-\phi(p)} \\
    \le& \begin{multlined}[t]
       \abs*{\Mean\bracket*{\frac{2p}{3n^2}\phi^{(3)}(p)-\frac{2\tilde{N}}{3n^3}H^{(3)}_{6,\Delta}[\phi]\paren*{\frac{\tilde{N}}{n}}}}
       +\abs*{\Mean\bracket*{\frac{7p}{24n^3}\phi^{(4)}(p)-\frac{7\tilde{N}}{24n^4}H^{(4)}_{6,\Delta}[\phi]\paren*{\frac{\tilde{N}}{n}}}} \\
       +\abs*{\Mean\bracket*{\frac{3p^2}{8n^2}\phi^{(4)}(p)-\frac{3\tilde{N}^2}{8n^4}H^{(4)}_{6,\Delta}[\phi]\paren*{\frac{\tilde{N}}{n}}}}
       + \frac{5p^2}{12n^3}\abs*{\phi^{(5)}(p)} + \frac{11p}{120n^4}\abs*{\phi^{(5)}(p)} \\
       + \abs*{\Mean\bracket*{R_5\paren*{\frac{\tilde{N}}{n};H_{6,\Delta},p}}} + \abs*{\Mean\bracket*{\frac{\tilde{N}}{2n^2}R_3\paren*{\frac{\tilde{N}}{n};H^{(2)}_{6,\Delta},p}}} .
   \end{multlined} \label{eq:plugin-bias1}
  \end{align}
  Let $\hat{p}=\frac{\tilde{N}}{n}$, and $G(x)=\frac{1}{x}(\hat{p}-x)^2$. For a doubly differentiable function $g$, the Taylor theorem and the mean value theorem give that there exists $\xi$ between $p$ and $\hat{p}$ such that
  \begin{align}
    \abs*{\Mean\bracket*{g(\hat{p})-g(p)}} =& \abs*{\Mean\bracket*{g^{(2)}(\xi)(\hat{p}-\xi)\frac{G(\hat{p})-G(p)}{G^{(1)}(\xi)}}} \\
    =& \abs*{\Mean\bracket*{g^{(2)}(\xi)\frac{\xi^2(\hat{p}-p)^2}{p(\xi+\hat{p})}}} \\
    \le& \frac{\sup_{\xi > 0}\xi\abs*{g^{(2)}(\xi)}}{n}.
  \end{align}
  For $g(p)=pH^{(3)}_{6,\Delta}[\phi](p)$, $\sup_{\xi > 0}\xi\abs*{g^{(2)}(\xi)} = \sup_{\xi > 0}\xi\abs*{2H^{(4)}_{6,\Delta}[\phi](\xi)+\xi H^{(5)}_{6,\Delta}[\phi](\xi)} \lesssim \Delta^{\alpha-3}$ because of \cref{lem:hermite-bound}. Thus, the first term in \cref{eq:plugin-var1} is bounded above as
  \begin{align}
    \abs*{\Mean\bracket*{\frac{2p}{3n^2}\phi^{(3)}(p)-\frac{2\tilde{N}}{3n^3}H^{(3)}_{6,\Delta}[\phi]\paren*{\frac{\tilde{N}}{n}}}} \lesssim \frac{\Delta^{\alpha-3}}{n^3}.
  \end{align}
  Similarly, the second and third terms in \cref{eq:plugin-var1} is bounded above as
  \begin{align}
    \abs*{\Mean\bracket*{\frac{7p}{24n^3}\phi^{(4)}(p)-\frac{7\tilde{N}}{24n^4}H^{(4)}_{6,\Delta}[\phi]\paren*{\frac{\tilde{N}}{n}}}} \lesssim& \frac{\Delta^{\alpha-4}}{n^4} \lesssim \frac{\Delta^{\alpha-3}}{n^3}, \\
    \abs*{\Mean\bracket*{\frac{3p^2}{8n^2}\phi^{(4)}(p)-\frac{3\tilde{N}^2}{8n^4}H^{(4)}_{6,\Delta}[\phi]\paren*{\frac{\tilde{N}}{n}}}} \lesssim& \frac{\Delta^{\alpha-3}}{n^3}.
  \end{align}
  From the divergence speed assumption, the fourth and fifth terms in \cref{eq:plugin-var1} is bounded above as
  \begin{align}
    \frac{5p^2}{12n^3}\abs*{\phi^{(5)}(p)} \lesssim \frac{p^{\alpha-3}}{n^3},
    \frac{11p}{120n^4}\abs*{\phi^{(5)}(p)} \lesssim \frac{p^{\alpha-4}}{n^3} \lesssim \frac{p^{\alpha-3}}{n^3}.
  \end{align}
  Let $G(x)=\frac{1}{x^3}(\hat{p}-x)^6$. For a sixth time differentiable function $g$, the Taylor theorem and the mean value theorem give that there exists $\xi$ between $p$ and $\hat{p}$ such that
 \begin{align}
   \abs*{\Mean\bracket*{R_5(\hat{p};g,p)}} =& \abs*{\Mean\bracket*{g^{(6)}(\xi)(\hat{p}-\xi)^5\frac{G(\hat{p})-G(p)}{G^{(1)}(\xi)}}} \\
   =& \abs*{\Mean\bracket*{g^{(6)}(\xi)\frac{\xi^4(\hat{p}-p)^6}{3p^3(\xi+\hat{p})}}} \\
   \le& \paren*{\frac{5}{n^3}+\frac{25}{3n^4p}+\frac{1}{3n^5p^2}}\sup_{\xi > 0}\xi^3\abs*{g^{(6)}(\xi)}.
 \end{align}
 For $g(p)=H_{6,\Delta}[\phi](p)$, $\sup_{\xi > 0}\xi^3\abs*{g^{(6)}(\xi)} = \sup_{\xi > 0}\xi^3\abs*{H^{(6)}_{6,\Delta}[\phi](\xi)} \lesssim \Delta^{\alpha-3}$. Hence, we have
 \begin{align}
   \abs*{\Mean\bracket*{R_5\paren*{\frac{\tilde{N}}{n};H_{6,\Delta},p}}} \lesssim& \frac{\Delta^{\alpha-3}}{n^3}+\frac{\Delta^{\alpha-3}}{n^4p}+\frac{\Delta^{\alpha-3}}{n^5p^2} \\
   \lesssim& \frac{\Delta^{\alpha-3}}{n}.
 \end{align}
 Let $G(x)=(\hat{p}-x)^4/x^3$. For a fourth time differentiable function $g$, the Taylor theorem and the mean value theorem give that there exists $\xi$ between $p$ and $\hat{p}$ such that
 \begin{align}
   \abs*{\Mean\bracket*{\hat{p}R_3(\hat{p};g,p)}} =& \abs*{\Mean\bracket*{\hat{p}g^{(4)}(\xi)(\hat{p}-\xi)^3\frac{G(\hat{p})-G(p)}{G^{(1)}(\xi)}}} \\
   =& \abs*{\Mean\bracket*{g^{(4)}(\xi)\frac{\xi^4\hat{p}(\hat{p}-p)^4}{p^3(\xi+3\hat{p})}}} \\
   \le& \paren*{\frac{3}{n^2}+\frac{11}{n^3p}+\frac{1}{n^4p^2}}\sup_{\xi > 0}\xi^3\abs*{g^{(4)}(\xi)}.
 \end{align}
 For $g(p)=H^{(2)}_{6,\Delta}[\phi](p)$, $\sup_{\xi > 0}\xi^3\abs*{g^{(4)}(\xi)} \lesssim \Delta^{\alpha-3}$. Hence,
 \begin{align}
   \abs*{\Mean\bracket*{\frac{\tilde{N}}{2n^2}R_3\paren*{\frac{\tilde{N}}{n};H^{(2)}_{6,\Delta},p}}} \lesssim \frac{\Delta^{\alpha-3}}{n^3}.
 \end{align}

\end{proof}
\begin{proof}[Proof of \cref{lem:plugin-var-1-3/2}]
  By the triangle inequality, we have
  \begin{align}
    & \frac{1}{5}\Var\bracket*{\bar\phi_{4,\Delta}\paren*{\frac{\tilde{N}}{n}}} \\
    \le& \frac{1}{5}\Mean\bracket*{\paren*{\bar\phi_{4,\Delta}\paren*{\frac{\tilde{N}}{n}} - \bar\phi_{4,\Delta}(p)}^2} \\
    \le& \begin{multlined}[t]
      \Mean\bracket*{\paren*{H_{6,\Delta}[\phi]\paren*{\frac{\tilde{N}}{n}} - H_{6,\Delta}[\phi](p)}^2} + \Mean\bracket*{\paren*{\frac{\tilde{N}}{2n^2}H^{(2)}_{6,\Delta}[\phi]\paren*{\frac{\tilde{N}}{n}} - \frac{p}{2n}H^{(2)}_{6,\Delta}[\phi](p)}^2} \\ + \Mean\bracket*{\paren*{\frac{2\tilde{N}}{3n^3}H^{(3)}_{6,\Delta}[\phi]\paren*{\frac{\tilde{N}}{n}} - \frac{2p}{3n^2}H^{(3)}_{6,\Delta}[\phi](p)}^2} \\ + \Mean\bracket*{\paren*{\frac{7\tilde{N}}{24n^4}H^{(4)}_{6,\Delta}[\phi]\paren*{\frac{\tilde{N}}{n}} - \frac{7p}{24n^3}H^{(4)}_{6,\Delta}[\phi](p)}^2} \\ + \Mean\bracket*{\paren*{\frac{3\tilde{N}^2}{8n^4}H^{(4)}_{6,\Delta}[\phi]\paren*{\frac{\tilde{N}}{n}} - \frac{3p^2}{8n^2}H^{(4)}_{6,\Delta}[\phi](p)}^2}. \label{eq:plugin-var1}
    \end{multlined}
  \end{align}
  We derive upper bounds for each term in \cref{eq:plugin-var1}. We derive a bound on the first term in \cref{eq:plugin-var1}. Applying \cref{lem:hermite-bound} with $\ell = 1 \le \alpha$, we obtain that $H_{6,\Delta}[\phi]$ is Lipschitz continuous. Hence,
  \begin{align}
    \Mean\bracket*{\paren*{H_{6,\Delta}[\phi]\paren*{\frac{\tilde{N}}{n}} - H_{6,\Delta}[\phi](p)}^2} \le& \sup_{\xi > 0}\Mean\bracket*{\paren*{H^{(1)}_{6,\Delta}[\phi](\xi)\paren*{\frac{\tilde{N}}{n} - p}}^2} \\
    \lesssim& \Var\bracket*{\frac{\tilde{N}}{n}} = \frac{p}{n}.
  \end{align}

  Next, we derive a bound on the second term in \cref{eq:plugin-var1}. Let $\hat{p}=\frac{\tilde{N}}{n}$, $g(p)=pH^{(2)}_{6,\Delta}[\phi](p)$, and $G(x)=\frac{1}{\sqrt{x}}(\hat{p}-x)$. The Taylor theorem and the mean value theorem give that there exists $\xi$ between $p$ and $\hat{p}$ such that
  \begin{align}
    \Mean\bracket*{\paren*{g(\hat{p})-g(p)}^2} =& \Mean\bracket*{\paren*{g^{(1)}(\xi)\frac{G(\hat{p})-G(p)}{G^{(1)}(\xi)}}^2} \\
    =& \Mean\bracket*{\paren*{g^{(1)}(\xi)\frac{2\xi^{\frac{3}{2}}(\hat{p}-p)}{\sqrt{p}(\xi+\hat{p})}}^2} \\
    \le& \frac{4\sup_{\xi > 0}\xi\abs*{g^{(1)}(\xi)}^2}{p}\Mean\bracket*{(\hat{p}-p)^2} \\
    =& \frac{4\sup_{\xi > 0}\xi\abs*{g(\xi)}^2}{n},
  \end{align}
  where we use $\Mean\bracket*{(\hat{p}-p)^2} = p/n$ to get the last line. From \cref{lem:hermite-bound}, we have $\sup_{\xi > 0}\xi\abs*{g^{(1)}(\xi)}^2 = \sup_{\xi > 0}\xi\abs*{H^{(2)}_{6,\Delta}[\phi](\xi)+\xi H^{(3)}_{6,\Delta}[\phi](\xi)}^2 \lesssim \Delta^{2\alpha-3}$. Hence,
  \begin{align}
    \Mean\bracket*{\paren*{\frac{\tilde{N}}{2n^2}H^{(2)}_{6,\Delta}[\phi]\paren*{\frac{\tilde{N}}{n}} - \frac{p}{2n}H^{(2)}_{6,\Delta}[\phi](p)}^2} \lesssim \frac{\Delta^{2\alpha-3}}{n^3} \lesssim \frac{1}{n^{2\alpha}}.
  \end{align}

  Bounds on the third, fourth, and fifth terms in \cref{eq:plugin-var1} can be obtained in the similar manner of the second term. Let $g(p)=pH^{(3)}_{6,\Delta}[\phi](p)$, $g(p)=pH^{(4)}_{6,\Delta}[\phi](p)$, or $g(p)=p^2H^{(4)}_{6,\Delta}[\phi](p)$. With the same manner of the second term, we have
  \begin{align}
    \Mean\bracket*{\paren*{g(\hat{p})-g(p)}^2} \le& \frac{4\sup_{\xi > 0}\xi\abs*{g^{(1)}(\xi)}^2}{n}.
  \end{align}
  \sloppy From \cref{lem:hermite-bound}, we have $\sup_{\xi > 0}\xi\abs*{H^{(3)}_{6,\Delta}[\phi](\xi)+\xi H^{(4)_{6,\Delta}[\phi](\xi)}}^2 \lesssim \Delta^{2\alpha-5}$, $\sup_{\xi > 0}\xi\abs*{H^{(4)}_{6,\Delta}[\phi](\xi)+\xi H^{(5)}_{6,\Delta}[\phi](\xi)}^2 \lesssim \Delta^{2\alpha-7}$, and $\sup_{\xi > 0}\xi\abs*{2\xi H^{(4)}_{6,\Delta}[\phi](\xi) + \xi^2H^{(5)}_{6,\Delta}[\phi](\xi)}^2 \lesssim \Delta^{2\alpha-5}$. Hence,
  \begin{align}
    \Mean\bracket*{\paren*{\frac{2\tilde{N}}{3n^3}H^{(3)}_{6,\Delta}[\phi]\paren*{\frac{\tilde{N}}{n}} - \frac{2p}{3n^2}H^{(3)}_{6,\Delta}[\phi](p)}^2} \lesssim& \frac{\Delta^{2\alpha-5}}{n^5} \lesssim \frac{1}{n^{2\alpha}}, \\
    \Mean\bracket*{\paren*{\frac{7\tilde{N}}{24n^4}H^{(4)}_{6,\Delta}[\phi]\paren*{\frac{\tilde{N}}{n}} - \frac{7p}{24n^3}H^{(4)}_{6,\Delta}[\phi](p)}^2} \lesssim& \frac{\Delta^{2\alpha-7}}{n^7} \lesssim \frac{1}{n^{2\alpha}}, \\
    \Mean\bracket*{\paren*{\frac{3\tilde{N}^2}{8n^4}H^{(4)}_{6,\Delta}[\phi]\paren*{\frac{\tilde{N}}{n}} - \frac{3p^2}{8n^2}H^{(4)}_{6,\Delta}[\phi](p)}^2} \lesssim& \frac{\Delta^{2\alpha-5}}{n^5} \lesssim \frac{1}{n^{2\alpha}}.
  \end{align}
\end{proof}

\section{Upper Bound Analysis for $\alpha \in (1,3/2)$}
Here, we prove the upper part of \cref{thm:optimal-rate-1-3/2}. In this section, we denote $\hat\theta$ as an estimator with the basic construction where $\phi_{\mathrm{poly}}$ is the best polynomial estimator and $\phi_{\mathrm{plugin}}$ is the fourth order bias corrected plugin estimator. We prove the following theorem.
\begin{theorem}\label{thm:upper-bound}
  Suppose $\Delta_{n,k} = C_2\ln n$ and $L = \floor{C_1\ln n}$ where $C_1$ and $C_2$ are universal constants such that $6C_1\ln 2 + 4\sqrt{C_1C_2}(1+\ln2) \le 3-2\alpha$ and $C_2 > 16\alpha$. If the sixth divergence speed of $\phi$ is $p^\alpha$ for $\alpha \in (1,3/2)$, the worst-case risk of $\hat\theta$ is bounded above by
  \begin{align}
    \sup_{P \in \dom{M}_k} \Mean\bracket*{\paren*{\hat\theta\paren{\tilde{N}} - \theta(P)}^2} \lesssim \frac{k^2}{(n\ln n)^{2\alpha}} + \frac{1}{n}.
  \end{align}
\end{theorem}
To prove \cref{thm:upper-bound}, we use the following bounds on the bias and the variance of $\hat\theta$.
\begin{lemma}[{\autocite[Lemma 2]{DBLP:journals/corr/FukuchiS17}}]\label{lem:upper-ind-bias}
  Given $P \in \dom{M}_k$, for $1 \lesssim \Delta_{n,k} \le n$, the bias of $\hat\theta$ is bounded above by
  \begin{multline}
   \Bias\bracket*{\hat\theta\paren{\tilde{N}} - \theta(P)} \lesssim
   \sum_{i = 1}^k \paren[\Bigg]{
     (e/4)^{\Delta_{n,k}} + \Bias\bracket*{\phi_{\mathrm{plugin}}(\tilde{N}_i) - \phi(p_i) }\ind{np_i > \Delta_{n,k}} \\
     + \Bias\bracket*{\phi_{\mathrm{poly}}(\tilde{N}_i) - \phi(p_i) }\ind{np_i \le 4\Delta_{n,k}} + e^{-\Delta_{n,k}/8}
   }.
 \end{multline}
\end{lemma}
\begin{lemma}[{\autocite[Lemma 3]{DBLP:journals/corr/FukuchiS17}}]\label{lem:upper-ind-var}
 Given $P \in \dom{M}_k$, for $1 \lesssim \Delta_{n,k} \le n$, the variance of $\hat\theta$ is bounded above by
 \begin{multline}
  \Var\bracket*{\hat\theta\paren{\tilde{N}} - \theta(P)} \lesssim
   \sum_{i = 1}^k \paren[\Bigg]{
   (e/4)^{\Delta_{n,k}} + \Var\bracket*{\phi_{\mathrm{plugin}}(\tilde{N}_i) - \phi(p_i) }\ind{np_i > \Delta_{n,k}} \\
   + \Var\bracket*{\phi_{\mathrm{poly}}(\tilde{N}_i) - \phi(p_i)  }\ind{np_i \le 4\Delta_{n,k}} + e^{-\Delta_{n,k}/8} \\ + \paren*{\Bias\bracket*{\phi_{\mathrm{plugin}}(\tilde{N}_i) - \phi(p_i)} + \Bias\bracket*{\phi_{\mathrm{poly}}(\tilde{N}_i) - \phi(p_i)}}^2\ind{\Delta_{n,k} \le p_i \le 4\Delta_{n,k}}
  }.
 \end{multline}
\end{lemma}
The proof of \cref{lem:upper-ind-bias,lem:upper-ind-var} is valid for the case $\alpha \in (1,3/2)$ if $\Bias\bracket{\phi_{\mathrm{plugin}}(\tilde{N}_i) - \phi(p_i) } \lesssim 1$ and $\Bias\bracket{\phi_{\mathrm{poly}}(\tilde{N}_i) - \phi(p_i)} \lesssim 1$. For $\phi_{\mathrm{plugin}}$, we can check from the last truncation and Chebyshev alternative theorem that $\Bias\bracket{\phi_{\mathrm{plugin}}(\tilde{N}_i) - \phi(p_i) } \lesssim 1 + E_L(\phi,[0,\Delta_{n,k}/n]) \lesssim 1$. For $\phi_{\mathrm{poly}}$, application of \cref{lem:hermite-bound} yields the claim.

\begin{proof}[Proof of \cref{thm:upper-bound}]
  Combine \cref{lem:upper-ind-var,lem:upper-ind-bias,lem:poly-var,lem:poly-bias,lem:plugin-var-1-3/2,lem:plugin-bias-1-3/2}. As long as $C_2(\ln 4 - 1) \ge 2\alpha$ and $C_2 \ge 16\alpha$, we have
  \begin{align}
    (e/4)^{\Delta_{n,k}} \lesssim& (n\ln n)^{-2\alpha}, \\
    e^{-\Delta_{n,k}/8} \lesssim& (n\ln n)^{-2\alpha}.
  \end{align}
  Besides, as long as $6C_1\ln 2+4\sqrt{C_1C_2}(\ln 2 + 1) \le 3-2\alpha$, we have
  \begin{align}
    \Bias\bracket*{\phi_{\mathrm{poly}}(\tilde{N}_i) - \phi(p_i)  }\ind{np_i \le 4\Delta_{n,k}} \lesssim& (n\ln n)^{-\alpha}, \\
    \Var\bracket*{\phi_{\mathrm{poly}}(\tilde{N}_i) - \phi(p_i)  }\ind{np_i \le 4\Delta_{n,k}} \lesssim& (n\ln n)^{-2\alpha}.
  \end{align}
  The remain terms are the bias and variance of $\phi_{\mathrm{plugin}}$. For the bias, we have from \cref{lem:plugin-bias-1-3/2} that for $\alpha \in (1,3/2)$,
  \begin{align}
    \Bias\bracket*{\phi_{\mathrm{plugin}}(\tilde{N}_i) - \phi(p_i) }
    \lesssim n^{-\alpha}\ln^{-(3-\alpha)}n \lesssim (n\ln n)^{-\alpha}.
  \end{align}
  For the variance, by \cref{lem:plugin-var-1-3/2}, we have
  \begin{align}
    &\sum_{i : np_i \ge \Delta_{n,k}}\Var\bracket*{\phi_{\mathrm{plugin}}(\tilde{N}_i) } \\
    \lesssim& \sum_{i : np_i \ge \Delta_{n,k}}\paren*{\frac{1}{n^{2\alpha}} + \frac{p_i}{n}} \\
    \le& \frac{1}{n^{2\alpha-1}\Delta_{n,k}} + \frac{1}{n} \lesssim \frac{1}{n}.
  \end{align}
\end{proof}

\section{Proof of \cref{lem:hermite-bound}}

\begin{proof}[Proof of \cref{lem:hermite-bound}]
  It is clear that $p^\beta\abs*{H^{(\ell)}_{L,\Delta}[\phi](p)} \lesssim p^{\alpha+\beta-\ell}$ for $p \in [\Delta,1]$ because of the divergence speed assumption. For $p \ge 2$, $\abs*{H^{(\ell)}_{L,\Delta}[\phi](p)} = 0$ by definition. For $p \in (1,2)$, we have
  \begin{align}
    & p^\beta\abs*{H^{(\ell)}_L\paren*{p;\phi,1,2}} \\
    =& \begin{multlined}[t]
      p^\beta\abs[\Bigg]{\sum_{u=0}^\ell\binom{\ell}{u}\sum_{m=1\lor u}^{L}\frac{\phi^{(m)}(1)}{m!}(p-1)^{m-u}\sum_{s=0}^{L-m}\frac{L+1}{L+s+1}\\\paren*{\prod_{w=0}^{\ell-u-1}(L+s+1-w)}\sum_{w=\ell-u-s}^{\ell-u}(-1)^w\binom{\ell-u}{w}\Beta_{s-\ell+u+w,L+s+1-\ell+u}\paren*{p-1}}
    \end{multlined} \\
    \le& \begin{multlined}[t]
      (L+1)p^\beta\sum_{u=0}^\ell\binom{\ell}{u}\sum_{m=1\lor u}^{L}\frac{\abs*{\phi^{(m)}(1)}}{m!}\sum_{s=0}^{L-m}\\\paren*{\prod_{w=1}^{\ell-u-1}(L+s+1-w)}\sum_{w=\ell-u-s}^{\ell-u}\binom{\ell-u}{w}\abs*{\Beta_{s-\ell+u+w,L+s+1-\ell+u}\paren*{p-1}}
    \end{multlined}\\
    \le& \begin{multlined}[t]
      (L+1)p^\beta\sum_{u=0}^\ell\binom{\ell}{u}\sum_{m=1\lor u}^{L}\frac{W_m+c_m}{m!}\sum_{s=0}^{L-m}\paren*{\prod_{w=1}^{\ell-u-1}(L+s+1-w)}\\\sum_{w=\ell-u-s}^{\ell-u}\binom{\ell-u}{w}\binom{L+s+1-\ell+u}{s-\ell+u+w}\paren*{\frac{s-\ell+u+w}{L+s+1-\ell+u}}^{s-\ell+u+w}\paren*{\frac{L+1-w}{L+s+1-\ell+u}}^{L+1-w}
    \end{multlined} \\
    \lesssim 1.
  \end{align}
  Similarly, for $p < \Delta$, we have
  \begin{align}
    & p^\beta\abs*{H^{(\ell)}_L\paren*{p;\phi,\Delta,\frac{\Delta}{2}}} \\
    \le& \begin{multlined}[t]
      (L+1)\sum_{u=0}^\ell\binom{\ell}{u}\sum_{m=1\lor u}^{L}\frac{\abs*{\phi^{(m)}(\Delta)}}{m!}\paren*{\Delta-p}^{m-u}\sum_{s=0}^{L-m}\paren*{\prod_{w=1}^{\ell-u-1}(L+s+1-w)}\\\sum_{w=\ell-u-s}^{\ell-u}\binom{\ell-u}{w}\binom{L+s+1-\ell+u}{s-\ell+u+w}\paren*{\frac{s-\ell+u+w}{L+s+1-\ell+u}}^{s-\ell+u+w}\paren*{\frac{L+1-w}{L+s+1-\ell+u}}^{L+1-w}
    \end{multlined} \\
    \le& \begin{multlined}[t]
      (L+1)\Delta^\beta\sum_{u=0}^\ell\binom{\ell}{u}\sum_{m=1\lor u}^{L}\frac{\paren*{W_m\Delta^{\alpha-m}+c_m}\Delta^{m-u}}{2^{m-u}m!}\sum_{s=0}^{L-m}\paren*{\prod_{w=1}^{\ell-u-1}(L+s+1-w)}\\\sum_{w=\ell-u-s}^{\ell-u}\binom{\ell-u}{w}\binom{L+s+1-\ell+u}{s-\ell+u+w}\paren*{\frac{s-\ell+u+w}{L+s+1-\ell+u}}^{s-\ell+u+w}\paren*{\frac{L+1-w}{L+s+1-\ell+u}}^{L+1-w}
    \end{multlined} \\
    \le& \begin{multlined}[t]
      (L+1)\Delta^\beta\sum_{u=0}^\ell\binom{\ell}{u}\sum_{m=1\lor u}^{L}\frac{W_m\Delta^{\alpha-u}+c_m}{2^{m-u}m!}\sum_{s=0}^{L-m}\paren*{\prod_{w=1}^{\ell-u-1}(L+s+1-w)}\\\sum_{w=\ell-u-s}^{\ell-u}\binom{\ell-u}{w}\binom{L+s+1-\ell+u}{s-\ell+u+w}\paren*{\frac{s-\ell+u+w}{L+s+1-\ell+u}}^{s-\ell+u+w}\paren*{\frac{L+1-w}{L+s+1-\ell+u}}^{L+1-w}
    \end{multlined} \\
    \lesssim& \Delta^{\alpha+\beta-\ell}.
  \end{align}
  Hence, $p^\beta\abs*{H^{(\ell)}_{L,\Delta}[\phi](p)} \lesssim 1\lor\Delta^{\alpha+\beta-\ell}$.
\end{proof}

\section{Proof of \cref{lem:first-moduli-bound}}
\begin{proof}[Proof of \cref{lem:first-moduli-bound}]
  From the divergence speed assumption, we have for $x,y \in (-1,1)$,
  \begin{align}
    &\abs*{\phi^{(2)}_\Delta(x)-\phi^{(2)}_\Delta(y)} \\
    \le& 2\Delta\abs*{\phi^{(1)}(\Delta x^2)-\phi^{(1)}(\Delta y^2)}+4\Delta^2\abs*{x^2\phi^{(2)}(\Delta x^2)-y^2\phi^{(2)}(\Delta y^2)} \\
    \le& 4\Delta^2\abs*{\int_y^xs\phi^{(2)}(\Delta s^2)ds} + 8\Delta^2\abs*{\int_y^x\paren*{s\phi^{(2)}(\Delta s^2)+\Delta s^3\phi^{(3)}(\Delta s^2)}ds} \\
    \le& 4\Delta^{\alpha}\abs*{\int_y^x(W_2s^{2\alpha-3}+sc_2)ds} + 8\Delta^{\alpha}\abs*{\int_y^x(W_2s^{2\alpha-3}+sc_2)ds} + 8\Delta^{\alpha}\abs*{\int_y^x(W_3s^{2\alpha-3}+s^2c_3)ds} \\
    \le& \Delta^{\alpha}\paren*{\frac{12W_2+8W_3}{2\alpha-2}\abs*{x^{2\alpha-2}-y^{2\alpha-2}} + (12c_2+8c_3)\abs*{x-y}}.
  \end{align}
  Since $x \to x^{2\alpha-2}$ is H\"older continuous if $\alpha \in (1,3/2)$, we get the claim.
\end{proof}

\section{Proof of \cref{thm:variance-lipschitz}}
We utilize the concentration result of the bounded difference.
\begin{theorem}[see e.g., \autocite{boucheron2013concentration}]
  Suppose that $X_1,...,X_n$ are independent random variables on $\dom{X}$. For a function $f:\dom{X}^n\to\RealSet$, suppose there exist universal constants $c_1,...,c_n$ such that for any $i \in [n]$,
  \begin{align}
    \sup_{x_1,...,x_n,x'_i}\abs*{f(x_1,...,x_i,...,x_n) - f(x_1,...,x'_i,...,x_n)} \le c_i.
  \end{align}
  Then,
  \begin{align}
    \Var\bracket*{f(X_1,...,X_n)} \le \frac{1}{4}\sum_{i=1}^nc_i^2.
  \end{align}
\end{theorem}
\begin{proof}[Proof of \cref{thm:variance-lipschitz}]
  Suppose a sample $X_j$ is changed from $X_j=i$ to $X_j=i'$. Then, the change of the histogram is $N_i \to N_i-1$ and $N_{i'} \to N_{i'}+1$. Hence, for $f(S_n) = \sum_{i=1}^k\phi(N_i/n)$, we have
  \begin{align}
    \sup_{N}\abs*{\phi\paren*{\frac{N_i}{n}}+\phi\paren*{\frac{N_{i'}}{n}} - \phi\paren*{\frac{N_i-1}{n}} - \phi\paren*{\frac{N_{i'}+1}{n}}} \le \frac{2}{n},
  \end{align}
  where we use the Lipschitz continuousness of $\phi$. Hence,
  \begin{align}
    \Var\bracket*{\sum_{i=1}^k\paren*{\frac{N_i}{n}}} \le \frac{n}{4}\paren*{\frac{2}{n}}^2 = \frac{1}{n}.
  \end{align}
\end{proof}
\section{Proof of \cref{thm:bias-lipschitz}}
\begin{proof}[Proof of \cref{thm:bias-lipschitz}]
  Application of the Taylor theorem yields there exists $\xi_1,...,\xi_k$ such that
  \begin{align}
    & \Bias\bracket*{\sum_{i=1}^k\paren*{\frac{N_i}{n}} - \theta(P)} \\
    =& \abs*{\Mean\bracket*{\sum_{i=1}^k\paren*{\phi^{(1)}(p_i)\paren*{\frac{N_i}{n}-p_i} + \frac{\phi^{(2)}(\xi_i)}{2}\paren*{\frac{N_i}{n}-p_i}^2}}} \\
    \le& \Mean\bracket*{\sum_{i=1}^k\frac{\abs*{\phi^{(2)}(\xi_i)}}{2}\paren*{\frac{N_i}{n}-p_i}^2}.
  \end{align}
  From the Lipschitz continuousness, we have $\sup_{p \in (0,1)}\abs*{\phi^{(2)}(p)} \le \norm{\phi^{(1)}}_{C^0,1}$. Hence,
  \begin{align}
    & \Bias\bracket*{\sum_{i=1}^k\paren*{\frac{N_i}{n}} - \theta(P)} \\
    \le& \frac{\norm{\phi^{(1)}}_{C^0,1}}{2}\Mean\bracket*{\sum_{i=1}^k\paren*{\frac{N_i}{n}-p_i}^2} \\
    =& \frac{\norm{\phi^{(1)}}_{C^0,1}}{2}\sum_{i=1}^k\frac{p_i(1-p_i)}{n} \\
    \le& \frac{\norm{\phi^{(1)}}_{C^0,1}}{2n}.
  \end{align}
\end{proof}
\section{Proof of \cref{thm:bias-plugin}}
\begin{proof}[Proof of \cref{thm:bias-plugin}]
  We divide the alphabets into two cases; $p_i \le 1/n$ and $p_i > 1/n$.

  {\bfseries Case $p_i \le 1/n$.} Since $\phi(0) = 0$, we have from the Taylor theorem that there exists $\xi_i$ between $\frac{N_i}{n}$ and $p_i$ such that
  \begin{align}
    & \abs*{\Mean\bracket*{\sum_{i : p_i \le 1/n}\paren*{\phi\paren*{\frac{N_i}{n}} - \phi(p_i)}}} \\
    \le& \sum_{i : p_i \le 1/n}\p\cbrace*{N_i = 0}\abs*{\phi(p_i)} + \abs*{\Mean\bracket*{\sum_{i : p_i \le 1/n}\frac{\phi^{(2)}(\xi_i)}{2}\paren*{\frac{N_i}{n} - p_i}^2 \middle| N_i > 0}} \\
    \le& \sum_{i : p_i \le 1/n}\p\cbrace*{N_i = 0}\abs*{\phi(p_i)} + \Mean\bracket*{\sum_{i : p_i \le 1/n}\abs*{\frac{\phi^{(2)}(\xi_i)}{2}}\paren*{\frac{N_i}{n} - p_i}^2 \middle| N_i > 0} \\
    \le& \sum_{i : p_i \le 1/n}\p\cbrace*{N_i = 0}\abs*{\phi(p_i)} + \Mean\bracket*{\sum_{i : p_i \le 1/n}\frac{W_2\xi_i^{\alpha-2}+c_2}{2}\paren*{\frac{N_i}{n} - p_i}^2 \middle| N_i > 0}  \\
    \le& \sum_{i : p_i \le 1/n}\p\cbrace*{N_i = 0}\abs*{\phi(p_i)} + \Mean\bracket*{\sum_{i : p_i \le 1/n}\frac{W_2p_i^{\alpha-2}+c_2}{2}\paren*{\frac{N_i}{n} - p_i}^2 \middle| N_i > 0} \\
    \lesssim& \sum_{i : p_i \le 1/n}\paren*{p_i^{\alpha-2}+1}\Mean\bracket*{\paren*{\frac{N_i}{n} - p_i}^2} \\
    \lesssim& \sum_{i : p_i \le 1/n}\paren*{p_i^{\alpha-2}+1}\frac{p_i}{n} \lesssim \frac{1}{n^{\alpha-1}}. \label{eq:bias-plugin1}
  \end{align}
  where we use $\abs*{\phi(p_i)} \lesssim p_i^\alpha \le p_i^2$ and $\frac{N_i}{n} \ge \frac{1}{n}$ if $N_i > 0$.

  {\bfseries Case $p_i > 1/n$.}
  Combining \cref{lem:bias-moduli,lem:moduli-bound}, we have
  \begin{align}
    & \abs*{\Mean\bracket*{\sum_{i : p_i > 1/n}\paren*{\phi\paren*{\frac{N_i}{n}} - \phi(p_i)}}} \\
    \le& \sum_{i : p_i > 1/n}\Bias\bracket*{\phi\paren*{\frac{N_i}{n}} - \phi(p_i)} \\
    \lesssim& \sum_{i : p_i > 1/n}\frac{p_i^{\alpha/2}}{n^{\alpha/2}}.
  \end{align}
  Since $\sup_{P \in \dom{M}_k}\sum_{i : p_i > 1/n}p_i^{\alpha/2} \lesssim n^{1-\alpha/2}$, we have
  \begin{align}
    & \abs*{\Mean\bracket*{\sum_{i : p_i > 1/n}\paren*{\phi\paren*{\frac{N_i}{n}} - \phi(p_i)}}} \\
    \lesssim& \frac{n^{1-\alpha/2}}{n^{\alpha/2}} = \frac{1}{n^{\alpha-1}}. \label{eq:bias-plugin2}
  \end{align}

  Combining \cref{eq:bias-plugin1,eq:bias-plugin2}, we have
  \begin{align}
    & \Bias\bracket*{\sum_i\paren*{\phi\paren*{\frac{N_i}{n}} - \phi(p_i)}} \\
    \le& \abs*{\Mean\bracket*{\sum_{i : p_i \le 1/n}\paren*{\phi\paren*{\frac{N_i}{n}} - \phi(p_i)}}} + \abs*{\Mean\bracket*{\sum_{i : p_i > 1/n}\paren*{\phi\paren*{\frac{N_i}{n}} - \phi(p_i)}}} \\
    \lesssim& \frac{1}{n^{\alpha-1}}.
  \end{align}
\end{proof}

\section{Proof of \cref{thm:lower1}}\label{Sec:proof-thm-lower1}
We use the LeCam's two point method~\autocite{LeCam:1986:AMS:20451}. Let $P$ and $Q$ be two probability vectors in $\dom{M}_k$. Then, the lower bound is given by
\begin{lemma}[\textcite{LeCam:1986:AMS:20451}]\label{lem:le-cam-two-point}
  The minimax lower bound is given as
  \begin{align}
    R^*(n,k;\phi) \ge \frac{1}{4}\paren*{\theta(P) - \theta(Q)}^2e^{-n\KL(P,Q)},
  \end{align}
  where $\KL$ denotes the KL divergence.
\end{lemma}
From \cref{lem:le-cam-two-point}, we want to appropriately choose $P$ and $Q$ that maximizes difference between $\theta(P)$ and $\theta(Q)$ with small KL divergence between $P$ and $Q$. Given $p$ and $q$, we define $P$ and $Q$ as
\begin{align}
  P =& \paren*{1-p,\frac{p}{k-1},...,\frac{p}{k-1}}, \\
  Q =& \paren*{1-q,\frac{q}{k-1},...,\frac{q}{k-1}}.
\end{align}
With this definition, we derive the upper bound on $\KL(P,Q)$ and the lower bound on $\paren{\theta(P) - \theta(Q)}^2$ and appropriately choose $p$ and $q$ to prove the minimax lower bound.
\begin{proof}[Proof of \cref{thm:lower1}]
  We first upper bound the KL divergence between $P$ and $Q$.
  \begin{align}
    \KL(P,Q) \le& \frac{1}{2}\ChiSquare(P,Q) \\
    =& \frac{(p-q)^2}{2(1-p)} + (k-1)\frac{\paren*{\frac{p}{k-1} - \frac{q}{k-1}}^2}{2\frac{p}{k-1}} \\
    =& \frac{(p-q)^2}{2(-p)} + \frac{(p-q)^2}{2p} = \frac{(p-q)^2}{2p(1-p)}.
  \end{align}
  Thus, we choose $p$ and $q$ that maximizes $\abs{\theta(P) - \theta(Q)}$ under constraint $(p-q)^2/2p(1-p) \lesssim 1/n$. Application of the Taylor theorem yields that there exist $\xi_1$ between $1-p$ and $1-q$ and $\xi_2$ between $p$ and $q$ such that
  \begin{align}
    \abs*{\theta(P) - \theta(Q)} =& \abs*{\phi(1-p) - \phi(1-q) + (k-1)\paren*{\phi\paren*{\frac{p}{k-1}} - \phi\paren*{\frac{q}{k-1}}}} \\
    =& \abs*{\phi^{(1)}(\xi_1)(q-p) + \phi^{(1)}\paren*{\frac{\xi_2}{k-1}}(p-q)} \\
    =& \abs*{\phi^{(1)}(\xi_1) - \phi^{(1)}\paren*{\frac{\xi_2}{k-1}}}\abs*{p - q}.
  \end{align}
  We can assume $\phi^{(1)}(0) = 0$ because for any $c \in \RealSet$, $\theta(P;\phi) = \theta(P;\phi_c)$ for $\phi_c(p) = \phi(p) + c(p-1/k)$. Letting $p_0 = 1\land(c'_2/W_2)^{1/(\alpha-2)}$, $\abs*{\phi^{(2)}(p)} \ge 0$ for $p \in (0,p_0]$. Thus, $\phi^{(2)}$ has same sign in $p \in (0,p_0]$. Hence, for sufficiently large $p$ and $q$ such that $1-p,1-q < p_0$,
  \begin{align}
    \abs*{\phi^{(1)}(\xi_1)} =& \abs*{\int_0^{\xi_1}\phi^{(2)}(s)ds} \\
    \ge& \int_0^{\xi_1}\paren*{W_2s^{\alpha-2}-c'_2}ds \\
    =& \frac{W_2}{\alpha-1}\xi_1^{\alpha-1}-c'_2\xi_1 > 0.
  \end{align}
  Also, we have
  \begin{align}
    \abs*{\phi^{(1)}\paren*{\frac{\xi_2}{k-1}}} \le& \int_0^{\frac{\xi_2}{k-1}}\abs*{\phi^{(2)}(s)}ds \\
    \le& \int_0^{\frac{\xi_2}{k-1}}\paren*{W_2s^{\alpha-2}+c_2}ds \\
    \le& \frac{W_2}{(\alpha-1)(k-1)^{\alpha-1}}\xi_2^{\alpha-1}+\frac{c_2}{k-1}\xi_2 \\
    \le& \frac{W_2}{(\alpha-1)(k-1)^{\alpha-1}}(p \lor q)^{\alpha-1}+\frac{c_2}{k-1}(p \lor q) = o(1).
  \end{align}
  Thus, from the inverse of the triangle inequality, we have
  \begin{align}
    \abs*{\theta(P) - \theta(Q)} \gtrsim \abs*{p - q}.
  \end{align}
  Let $p$ be an sufficiently large universal constant such that $1-p < p_0$, and let $q=p+\frac{1}{\sqrt{n}}$. Then, $1-q < p_0$ and $(p-q)^2/2p(1-p) \lesssim 1/n$. Hence,
  \begin{align}
    R^*(n,k;\phi) \gtrsim \paren*{p - p + \frac{1}{\sqrt{n}}}^2 = \frac{1}{n}.
  \end{align}
\end{proof}

\section{Detailed Analysis of Lower Bound for $\alpha \in (1,3/2)$}
First, we derive the association between the minimax risk and the approximated minimax risk defined below. For $\epsilon \in (0,1)$, define the approximated probabilities as
\begin{align}
  \dom{M}_k(\epsilon) = \cbrace*{(p_1,...,p_k) \in \RealSet^k_+ : \abs*{\sum_{i=1}^kp_i - 1} \le \epsilon}.
\end{align}
With this definition, we define the approximated minimax risk as
\begin{align}
 \tilde{R}^*(n,k,\epsilon;\phi) = \inf_{\hat\theta}\sup_{P \in \dom{M}_k(\epsilon)}\Mean\bracket*{\paren*{\theta(P) - \hat\theta(\tilde{N})}^2}. \label{eq:approximated-minimax-risk}
\end{align}
Then, we get the following theorem.
\begin{theorem}\label{thm:lower-approximated}
  For $\bar\alpha \in (0,1]$, if $\phi$ is $\bar\alpha$-H\"older continuous, for any $k,n \in \NaturalSet$ and any $\epsilon < 1/3$,
  \begin{align}
    \tilde{R}^*(n/2,k;\phi) \ge \frac{1}{2}\tilde{R}^*(n,k,\epsilon;\phi) - 4\norm*{\phi}_{C^{0,\bar\alpha}}^2k^{2-2\bar\alpha}e^{-n/32} - \norm*{\phi}_{C^{0,\bar\alpha}}^2k^{2-2\bar\alpha}\epsilon^{2\bar\alpha}. \label{eq:approximated-lower}
  \end{align}
\end{theorem}
The lower bound of $\tilde{R}^*(n,k,\epsilon;\phi)$ is given by the following theorems.
\begin{theorem}\label{thm:approx-tv-lower}
  Let $U$ and $U'$ be random variables such that $U,U' \in [0,\lambda]$ and $\Mean[U]=\Mean[U']\le 1$ and $k\abs*{\Mean[\phi(U/k) - \phi(U'/k)]} \ge d$, where $\lambda \le k$. Let $\epsilon = 4\lambda/\sqrt{k}$. For $\alpha \in (1,2)$, if $\phi$ is Lipschitz continuous, and $\phi^{(1)}$ is $(\alpha-1)$-H\"older continuous,
  \begin{align}
    \tilde{R}^*(n,k,\epsilon;\phi) \ge \frac{d^2}{16}\paren*{\frac{7}{8} - k\TV\paren*{\Mean[\Poi(nU/k)], \Mean[\Poi(nU'/k)]} - \frac{32\norm*{\phi^{(1)}}_{C^{0,\alpha-1}}^2\lambda^{2\alpha}}{k^{2\alpha-1}d^2}}. \label{eq:approximated-tv-lower}
  \end{align}
\end{theorem}
To derive the upper bound of $\TV\paren*{\Mean[\Poi(nU/k)], \Mean[\Poi(nU'/k)]}$, we apply the following lemma proved by \textcite{wu2016minimax}.
\begin{lemma}[{\textcite[Lemma 3]{wu2016minimax}}]\label{lem:tv-poi-bound}
  Let $V$ and $V'$ be random variables on $[0,M]$. If $\Mean[V^j] = \Mean[V'^j]$, $j = 1,...,L$ and $L > 2eM$, then
  \begin{align}
    \TV(\Mean[\Poi(V)], \Mean[\Poi(V')]) \le \paren*{\frac{2eM}{L}}^L.
  \end{align}
\end{lemma}
Combining \cref{thm:lower-approximated,thm:approx-tv-lower,lem:tv-poi-bound} gives the following corollary.
\begin{corollary}\label{cor:moment-match-lower}
  For $\alpha \in (1,2)$, suppose $\phi$ is Lipschitz continuous, and $\phi^{(1)}$ is $(\alpha-1)$-H\"older continuous. Let $U$ and $U'$ be random variables such that
  \begin{enumerate}
    \item $U,U' \in [0,\lambda]$,
    \item $\Mean[U]=\Mean[U']\le 1$,
    \item $\Mean[U^m]=\Mean[U^m]$ for $m=0,...,L$ where $kL > 2en\lambda$, and
    \item $\abs*{\Mean\bracket{k\phi(U/k) - k\phi(U'/k)}} \ge d$ where $\lambda \le k$.
  \end{enumerate}
  As long as $\lambda/\sqrt{k} < 1/12$, we have
 \begin{multline}
   \tilde{R}^*(n/2,k;\phi) \ge \frac{d^2}{32}\paren[\Bigg]{\frac{7}{8} - k\paren*{\frac{2en\lambda}{kL}}^L - \frac{32\norm*{\phi^{(1)}}_{C^{0,\alpha-1}}^2\lambda^{2\alpha}}{k^{2\alpha-1}d^2} \\
   - \frac{128\norm*{\phi}_{C^{0,1}}^2e^{-n/32}}{d^2} - \frac{512\norm*{\phi}_{C^{0,1}}^2\lambda^2}{kd^2}}. \label{eq:moment-match-lower}
 \end{multline}
\end{corollary}
Let us construct the distribution of $U$ and $U'$. To this end, we use the following a pair of the probability measures.
\begin{lemma}\label{lem:prior-construction}
  For any given integer $L > 0$, there exists two probability measures $\nu_0$ and $\nu_1$ on $[\eta,1]$ such that
  \begin{gather}
    \Mean_{X \sim \nu_0}[X^m] = \Mean_{X \sim \nu_1}[X^m], \for m=0,...,L, \\
    \Mean_{X \sim \nu_0}[\phi(X)] - \Mean_{X \sim \nu_1}[\phi(X)] = 2E_L(\phi, [\eta,1]).
  \end{gather}
\end{lemma}
\begin{proof}[\cref{lem:prior-construction}]
  The proof is almost same as the proof of \textcite[Lemma 10]{jiao2015minimax}. It follows directly from a standard functional analysis argument proposed by \textcite{lepski1999estimation}. It suffices to replace $x^\alpha$ with $\phi(x)$ and $[0,1]$ with $[\eta,1]$ in the proof of \autocite[Lemma 1]{cai2011testing}.
\end{proof}
Using \cref{lem:prior-construction}, we can construct the following pair of the probability measures.
\begin{lemma}\label{lem:best-approx-solution}
  Suppose $\phi:[0,1]\to\RealSet$ be a function such that $\phi(0) = 0$. Define $\phi^\star(p) = \phi(x)/x$. For any given integer $L > 0$, $\eta > 0$, and $\gamma \in (0,1)$ such that $\gamma \le \eta$, there exists two probability measures $\nu_0$ and $\nu_1$ on $[0,\gamma/\eta]$ such that
  \begin{gather}
    \Mean_{X \sim \nu_0}[X] = \Mean_{X \sim \nu_1}[X] = \gamma, \\
    \Mean_{X \sim \nu_0}[X^m] = \Mean_{X \sim \nu_1}[X^m], \for m=2,...,L+1, \\
    \Mean_{X \sim \nu_0}[\phi(X)] - \Mean_{X \sim \nu_1}[\phi(X)] = 2\gamma E_L(\phi^\star, [\gamma,\gamma/\eta]).
  \end{gather}
\end{lemma}
As proved \cref{lem:best-approx-solution}, we get the following consequential corollary of \cref{cor:moment-match-lower,lem:best-approx-solution}.
\begin{corollary}\label{cor:poly-approx-lower}
  For $\alpha \in (1,2)$, suppose $\phi$ is Lipschitz continuous, and $\phi^{(1)}$ is $(\alpha-1)$-H\"older continuous. For any given integer $L > 0$, $\lambda \le k$, and $\gamma \in (0,1)$ such that $\gamma \le \lambda/k$, \cref{eq:moment-match-lower} holds with $d = 2k\gamma E_L(\phi^\star, [\gamma,\lambda/k])$ as long as $\lambda/\sqrt{k} < 1/12$.
\end{corollary}
We can get the minimax lower bound by deriving a lower bound on $E_L(\phi^\star, [\gamma,\lambda/k])$. Hence, substituting \cref{thm:lower-poly-approx} into \cref{cor:poly-approx-lower} yields the following corollary.
\begin{corollary}\label{cor:combined-lower}
  For $\alpha \in (1,2)$, suppose the second divergence speed of $\phi$ is $p^\alpha$. For any given integer $L > 0$ and $\gamma \in (0,1)$ such that $2L^2\gamma \le 1$, with $d \gtrsim 2k\gamma^\alpha$ and $2\sqrt{k}L^2\gamma < 1/12$, we have
 \begin{multline}
   \tilde{R}^*(n/2,k;\phi) \ge \frac{d^2}{32}\paren[\Bigg]{\frac{7}{8} - k\paren*{4enL\gamma}^L - \frac{128\norm*{\phi^{(1)}}_{C^{0,\alpha-1}}^2kL^{4\alpha}\gamma^{2\alpha}}{d^2} \\
   - \frac{128\norm*{\phi}_{C^{0,1}}^2e^{-n/32}}{d^2} - \frac{2048\norm*{\phi}_{C^{0,1}}^2kL^4\gamma^2}{d^2}}.
 \end{multline}
\end{corollary}
Combining these results, we prove \cref{thm:lower2}.
\begin{proof}[Proof of \cref{thm:lower2}]
  Let $d \ge 2ck\gamma^\alpha$. Then, we have
  \begin{multline}
    \tilde{R}^*(n/2,k;\phi) \ge \frac{c^2k^2\gamma^{2\alpha}}{8}\paren[\Bigg]{\frac{7}{8} - k\paren*{4enL\gamma}^L - \frac{32\norm*{\phi^{(1)}}_{C^{0,\alpha-1}}^2L^{4\alpha}}{c^2k} \\
    - \frac{32\norm*{\phi}_{C^{0,1}}^2e^{-n/32}}{c^2k^2\gamma^{2\alpha}} - \frac{512\norm*{\phi}_{C^{0,1}}^2L^4\gamma^{2-2\alpha}}{c^2k}}.
  \end{multline}
  It is sufficient to prove the claim in the case $k^2 \gtrsim n^{2\alpha-1}\ln^{2\alpha}n$ because of \cref{thm:lower1}. Let $k^2 \ge c' n^{2\alpha-1}\ln^{2\alpha}n$ for an universal constant $c' > 0$. Fix $\delta > 0$. Set $\gamma = \frac{C_1}{n^{1+\delta}\ln n}$ and $L=\ceil{C_2\ln n}$ where $C_1$ and $C_2$ are universal constants such that $4enL\gamma < 1$, $C_1$ is sufficiently small to satisfy \cref{thm:lower-poly-approx}, and $C_2$ is sufficiently large such that $\delta L > \alpha$ for $n > 1$. For sufficiently small $C_1 > 0$, we can check that $2\sqrt{k}L^2\gamma < 1/12$ for $\alpha \in (1,2)$ because $2\sqrt{k}L^2\gamma \le C_1O(n^{\alpha/2-1}\ln^{\alpha/2+1} n)$. Letting $\delta L = \alpha + \delta'$ for $\delta' > 0$, we have
  \begin{align}
    k\paren*{4enL\gamma}^L  =& k\paren*{\frac{4eC_1\ceil{C_2\ln n}^2}{n^\delta\ln n}}^{\ceil{C_2\ln n}} \\
    \lesssim& \frac{k\ln n}{n^{\alpha+\delta'}} = o(1) \because k \lesssim (n\ln n)^\alpha.
  \end{align}
  Also, for sufficiently small $\delta$, we have
  \begin{align}
    \frac{32\norm*{\phi^{(1)}}_{C^{0,\alpha-1}}^2L^{4\alpha}}{c^2k} \le& \frac{32\norm*{\phi^{(1)}}_{C^{0,\alpha-1}}^2\ceil{C_2\ln n}^{4\alpha}}{c^2\sqrt{c'}n^{2\alpha-1}\ln^{2\alpha}n} = o(1), \label{eq:second-o1} \\
    \frac{32\norm*{\phi}_{C^{0,1}}^2e^{-n/32}}{c^2k^2\gamma^{2\alpha}} \le& \frac{32\norm*{\phi}_{C^{0,1}}^2n^{2\alpha\delta+1}e^{-n/32}}{c^2c^{'2}C_1^{2\alpha}} = o(1), \label{eq:third-o1} \\
    \frac{512\norm*{\phi}_{C^{0,1}}^2\ceil{C_2\ln n}^4\gamma^{2-2\alpha}}{c^2k} \le& \frac{512\norm*{\phi}_{C^{0,1}}^2\ceil{C_2\ln n}^4C_1^{2-2\alpha}n^{\alpha-\frac{3}{2}+\delta(2\alpha-2)}\ln^{\alpha-2}n}{c^2\sqrt{c'}} = o(1), \label{eq:forth-o1}
  \end{align}
  as long as $\alpha \in (1,3/2)$.
  Note that \cref{eq:second-o1,eq:third-o1,eq:forth-o1} satisfy for any $C_2 < \infty$. Hence, for $\alpha \in (1,3/2)$,
  \begin{align}
    \tilde{R}^*(n/2,k;\phi) \gtrsim \frac{k^2}{(n^{1+\delta}\ln n)^\alpha}.
  \end{align}
  From arbitrariness of $\delta > 0$, we get the claim.
\end{proof}

\section{Proof of \cref{thm:lower-approximated}}
\begin{proof}[Proof of \cref{thm:lower-approximated}]
  This proof is following the same manner of the proof of \autocite[Lemma 1]{wu2016minimax}. Fix $\delta > 0$. Let $\hat\theta(\cdot,n)$ be a near-minimax optimal estimator for fixed sample size $n$, i.e.,
  \begin{align}
    \sup_{P \in \dom{M}_k}\Mean\bracket*{\paren*{\hat\theta(N,n) - \theta(P)}} \le \delta + R^*(k,n;\phi).
  \end{align}
  For an arbitrary approximate distribution $P \in \dom{M}_k$, we construct an estimator
  \begin{align}
    \tilde\theta(\tilde{N}) = \hat\theta(\tilde{N}, n'),
  \end{align}
  where $\tilde{N}_i \sim \Poi(np_i)$ and $n' = \sum_i\tilde{N}_i$. From the triangle inequality, $\bar\alpha$-H\"older continuousness of $\phi$, and \cref{lem:bound-sum-alpha}, we have
  \begin{align}
    & \frac{1}{2}\paren*{\tilde\theta(\tilde{N}) - \theta(P)}^2 \\
    \le& \frac{1}{2}\paren*{\abs*{\tilde\theta(\tilde{N}) - \theta\paren*{\frac{P}{\sum_ip_i}}} + \abs*{\theta\paren*{\frac{P}{\sum_ip_i}} - \theta(P)}}^2 \\
    \le& \frac{1}{2}\paren*{\abs*{\tilde\theta(\tilde{N}) - \theta\paren*{\frac{P}{\sum_ip_i}}} + \abs*{\norm*{\phi}_{C^{0,\bar\alpha}}\sum_i\abs*{\frac{p_i}{\sum_ip_i} - p_i}^{\bar\alpha}}}^2 \\
    \le& \frac{1}{2}\paren*{\abs*{\tilde\theta(\tilde{N}) - \theta\paren*{\frac{P}{\sum_ip_i}}} + \norm*{\phi}_{C^{0,\bar\alpha}}\epsilon^{\bar\alpha}\sum_i\abs*{\frac{p_i}{\sum_ip_i}}^{\bar\alpha}}^2 \\
    \le& \frac{1}{2}\paren*{\abs*{\tilde\theta(\tilde{N}) - \theta\paren*{\frac{P}{\sum_ip_i}}} + \norm*{\phi}_{C^{0,\bar\alpha}}k^{1-\bar\alpha}\epsilon^{\bar\alpha}}^2 \\
    \le& \paren*{\tilde\theta(\tilde{N}) - \theta\paren*{\frac{P}{\sum_ip_i}}}^2 + \norm*{\phi}^2_{C^{0,\bar\alpha}}k^{2-2\bar\alpha}\epsilon^{2\bar\alpha}.
  \end{align}
  For the first term, we observe that $\tilde{N} \sim \Mul(m, \tfrac{P}{\sum p_i})$ conditioned on $n' = m$. Therefore, we have
  \begin{align}
    \Mean\paren*{\tilde\theta(\tilde{N}) - \theta\paren*{\frac{P}{\sum_{i=1}^k p_i}}}^2
     =& \sum_{m = 0}^\infty \Mean\bracket*{\paren*{\tilde\theta(\tilde{N},m) - \theta\paren*{\frac{P}{\sum_{i=1}^k p_i}}}^2 \middle| n' = m}\p\cbrace{n' = m} \\
     \le& \sum_{m = 0}^\infty \tilde{R}^*(m,k;\phi)\p\cbrace{n' = m} + \delta.
  \end{align}
  From $\alpha$-H\"older continuousness of $\phi$ and \cref{lem:bound-sum-alpha}, we have
  \begin{align}
    \tilde{R}^*(m,k;\phi) \le& \sup_{P,P' \in \dom{M}_k}\paren*{\theta(P)-\theta(P')}^2 \\
     \le& \norm*{\phi}^2_{C^{0,\bar\alpha}}\sup_{P,P' \in \dom{M}_k}\paren*{\sum_i\abs*{p_i - p'_i}^{\bar\alpha}}^2 \\
     \le& 4\norm*{\phi}^2_{C^{0,\bar\alpha}}\sup_{P \in \dom{M}_k}\paren*{\sum_ip_i^{\bar\alpha}}^2 \\
     \le& 4\norm*{\phi}^2_{C^{0,\bar\alpha}}k^{2-2\bar\alpha}.
  \end{align}
  Note that $\tilde{R}^*(m,k;\phi)$ is a non-increasing function with respect to $m$. Since $n' \sim \Poi(n\sum_i p_i)$ and $\abs*{\sum_i p_i - 1} \le \epsilon \le 1/3$, applying Chernoff bound yields $\p\cbrace*{n' \le n/2} \le e^{-n/32}$. Thus, we have
  \begin{align}
    & \Mean\paren*{\tilde\theta(\tilde{N}) - \theta\paren*{\frac{P}{\sum_{i=1}^k p_i}}}^2 \\
     \le& \sum_{m \ge n/2} \tilde{R}^*(m,k;\phi)\p\cbrace{n' = m} + 4\norm*{\phi}^2_{C^{0,\bar\alpha}}k^{2-2\bar\alpha}\p\cbrace*{n' \le n/2} + \delta \\
     \le& \tilde{R}^*(n/K,k;\phi) + 4\norm*{\phi}^2_{C^{0,\bar\alpha}}k^{2-2\bar\alpha}e^{-n/32} + \delta.
  \end{align}
  The arbitrariness of $\delta$ gives the desired result.
\end{proof}

\section{Proof of \cref{thm:approx-tv-lower}}

\begin{proof}[Proof of \cref{thm:approx-tv-lower}]
  The proof follows the same manner of the proof of \autocite[Lemma 2]{wu2016minimax} expect \cref{eq:approx-tv-expect} below. Let $\beta = \Mean[U] = \Mean[U'] \le 1$. Define two random vectors
  \begin{align}
    P = \paren*{\frac{U_1}{k},...,\frac{U_k}{k}, 1-\beta}, P' = \paren*{\frac{U'_1}{k},...,\frac{U'_k}{k}, 1-\beta},
  \end{align}
  where $U_i$a nd $U'_i$ are independent copies of $U$ and $U'$, respectively. Put $\epsilon = 4\lambda/\sqrt{k}$. Define the two events:
  \begin{align}
    \event =& \bracket*{\abs*{\sum_i\frac{U_i}{k} - \beta} \le \epsilon, \abs*{\theta(P) - \Mean[\theta(P)]} \le d/4}, \\
    \event' =& \bracket*{\abs*{\sum_i\frac{U'_i}{k} - \beta} \le \epsilon, \abs*{\theta(P') - \Mean[\theta(P')]} \le d/4}.
  \end{align}
  Applying Chebyshev's inequality, the union bound, the triangle inequality and $\alpha$-H\"older continuousness give
  \begin{align}
    \p\event^c \le& \p\cbrace*{\abs*{\sum_i\frac{U_i}{k} - \beta} > \epsilon} + \p\cbrace*{\abs*{\theta(P) - \Mean[\theta(P)]} > d/4} \\
    \le& \frac{\Var[U]}{k\epsilon^2} + \frac{16\sum_{i}\Var[\phi(U_i/k)]}{d^2} \\
    \le& \frac{1}{16} + \frac{16\sum_i\Mean\bracket*{\paren*{\phi(U_i/k) - \phi(0)}^2}}{d^2}.
  \end{align}
  Without loss of generality, we can assume $\phi^{(1)}(0) = 0$ because $\theta(P;\phi) = \theta(P;\phi_c)$ for any $c \in \RealSet$ where $\phi_c(p) = \phi(p) + c(p-1/k)$. Hence, the Taylor theorem indicates that there exists $\xi_i$ between $0$ and $U_i/k$ such that
  \begin{align}
    \p\event^c \le& \frac{1}{16} + \frac{16\sum_i\Mean\bracket*{\paren*{U_i\paren*{\phi^{(1)}(\xi_i) - \phi^{(1)}(0)}/k}^2}}{d^2}.
  \end{align}
  From H\"older continuousness of $\phi^{(1)}$, we obtain
  \begin{align}
    \p\event^c \le& \frac{1}{16} + \frac{16\sum_i\Mean\bracket*{\norm{\phi^{(1)}}_{C^{0,\alpha-1}}U_i\xi^{\alpha-1}/k}^2}{d^2} \\
    \le& \frac{1}{16} + \frac{16\norm{\phi^{(1)}}^2_{C^{0,\alpha-1}}\lambda^{\alpha}}{k^{\alpha-1}d^2}\label{eq:approx-tv-expect}.
  \end{align}
  By the same manner, we have
  \begin{align}
    \p\event'^c \le \frac{1}{16} + \frac{16\norm*{\phi}_{C^{0,\alpha-1}}^2\lambda^{2\alpha}}{k^{2\alpha-1}d^2}.
  \end{align}
  We define two priors on the set $\dom{M}_k(\epsilon)$, the conditional distributions $\pi = P_{U|\event}$ and $\pi' = P_{U'|\event'}$. By the definition of events $\event,\event'$ and triangle inequality, we obtain that under $\pi,\pi'$,
  \begin{align}
    \abs*{\theta(P) - \theta(P')} \ge \frac{d}{2}.
  \end{align}
  By triangle inequality, we have the total variation of observations under $\pi,\pi'$ as
  \begin{align}
    \TV(P_{\tilde{N}|\event}, P_{\tilde{N}'|\event'})
     \le& \TV(P_{\tilde{N}|\event}, P_{\tilde{N}}) + \TV(P_{\tilde{N}}, P_{\tilde{N}'}) + \TV(P_{\tilde{N}'}, P_{\tilde{N}'|\event'}) \\
     =& \p\event^c + \TV(P_{\tilde{N}}, P_{\tilde{N}'}) + \p\event'^c \\
     \le& \TV(P_{\tilde{N}}, P_{\tilde{N}'}) + \frac{1}{8} + \frac{32\norm*{\phi}_{C^{0,\alpha-1}}^2\lambda^{2\alpha}}{k^{2\alpha-1}d^2}.
  \end{align}
  From the fact that total variation of product distribution can be upper bounded by the summation of individual ones, we obtain
  \begin{align}
    \TV(P_{\tilde{N}}, P_{\tilde{N}'})
     \le& \sum_{i=1}^k\TV(P_{\tilde{N}_i}, P_{\tilde{N}'_i}) + \TV(n(1-\beta), n(1-\beta)) \\
     =& k\TV(\Mean[\Poi(nU/k)],\Mean[\Poi(nU'/k)]).
  \end{align}
  Then, applying Le Cam's lemma~\autocite{LeCam:1986:AMS:20451} yields that
  \begin{align}
    \tilde{R}^*(n,k,\epsilon;\phi) \ge \frac{d^2}{16}\paren*{\frac{7}{8} - k\TV\paren*{\Mean[\Poi(nU/k)], \Mean[\Poi(nU'/k)]} - \frac{32\norm*{\phi}_{C^{0,\alpha-1}}^2\lambda^{2\alpha}}{k^{2\alpha-1}d^2}}.
  \end{align}
\end{proof}

\section{Proof of \cref{lem:best-approx-solution}}
\begin{proof}[Proof of \cref{lem:best-approx-solution}]
  From \cref{lem:prior-construction}, there exists a pair of probability measures $\rho_0$ and $\rho_1$ on $[\eta,1]$ such that
  \begin{gather}
    \Mean_{X \sim \rho_0}[X^m] = \Mean_{X \sim \rho_1}[X^m], \for m=0,...,L, \\
    \Mean_{X \sim \rho_0}[\phi(X)] - \Mean_{X \sim \rho_1}[\phi(X)] = 2E_L(\phi, [\eta,1]).
  \end{gather}
  For $X \sim \rho_i$, let $\rho'_i$ be probability measures such that $\gamma X/\eta \sim \rho'_i$. Define $\nu_0$ and $\nu_1$ such that
  \begin{align}
    \frac{d\nu_i}{d\rho_i}(u) = \frac{\gamma}{u} \textand \nu_i(\cbrace*{0}) = 1-\nu_i([\eta,1]).
  \end{align}
  By construction, $\nu_i$ are defined on $[0,\gamma/\eta]$. Besides, we get
  \begin{align}
    \Mean_{X \sim \nu_i}[X] =& \gamma\Mean_{X \sim \rho'_i}\bracket*{1} = \gamma, \\
    \Mean_{X \sim \nu_i}[X^m] =& \gamma\Mean_{X \sim \rho'_i}\bracket*{X^{m-1}} = \frac{\gamma^m}{\eta^{m-1}}\Mean_{X \sim \rho_i}\bracket*{X^{m-1}} \for m=2,...,L+1.
  \end{align}
  From the assumption $\phi(0)=0$, we have
  \begin{align}
    & \Mean_{X \sim \nu_0}[\phi(X)] - \Mean_{X \sim \nu_1}[\phi(X)] \\
    =& \gamma\paren*{\Mean_{X \sim \rho'_0}\bracket*{\frac{\phi(X)}{X}} - \Mean_{X \sim \rho'_1}\bracket*{\frac{\phi(X)}{X}}} \\
    =& \gamma\paren*{\Mean_{X \sim \rho_0}\bracket*{\frac{\phi(\gamma X/\eta)}{\gamma X/\eta}} - \Mean_{X \sim \rho_1}\bracket*{\frac{\phi(\gamma X/\eta)}{\gamma X/\eta}}} \\
    =& 2\gamma E_L\paren*{\varphi,[\gamma,\gamma/\eta]}.
  \end{align}
\end{proof}

\section{Proof of \cref{thm:lower-poly-approx}}

\begin{proof}[Proof of \cref{thm:lower-poly-approx}]
  Letting $\phi^\star_{\eta,\gamma}(x) = \phi^\star(\gamma\frac{1+\eta+(1-\eta)x}{2\eta})$, we have $E_L(\phi^\star,[\gamma,\gamma/\eta]) = E_L(\phi^\star_{\eta,\gamma},[-1,1])$. We utilize the first-order Ditzian-Totik modulus of smoothness~\autocite{ditzian2012moduli} defined as
  \begin{align}
    \omega^1_\varphi(f,t) = \sup_{x,y \in [-1,1]}\cbrace*{\abs*{f(x)-f(y)} : \abs*{x-y} \le 2t\varphi\paren*{\frac{x+y}{2}}},
  \end{align}
  where $\varphi(x)=\sqrt{1-x^2}$. Fix $y=-1$. For $t > 0$, we have
  \begin{gather}
    \abs*{x - y} \le 2t\varphi\paren*{\frac{x+y}{2}} \iff \\
    -1 \le x \le -1 + \frac{4}{t^{-2}+1}.
  \end{gather}
  For $t \in (0,1)$, $\frac{2}{t^{-2}+1} \ge t^2$. Hence,
  \begin{align}
    \omega^1_\varphi(\phi^\star_{\eta,\gamma},t) \ge& \sup_{x \in [-1,1]}\cbrace*{\abs*{\phi^\star_{\eta,\gamma}(x)-\phi^\star_{\eta,\gamma}(-1)} : 0 \le 1 + x \le \frac{4}{t^{-2}+1}} \\
    =& \sup_{x \in [0,1]}\cbrace[\Bigg]{\abs*{\phi^\star\paren*{\gamma\paren*{1+\frac{(1-\eta)x}{\eta}}}-\phi^\star(\gamma)} : 0 \le x \le t^2} \\
    =& \frac{1}{\gamma}\sup_{x \in [0,1]}\cbrace*{\frac{1}{1+\frac{(1-\eta)x}{\eta}}\abs*{\phi\paren*{\gamma\paren*{1+\frac{(1-\eta)x}{\eta}}}-\phi(\gamma)\paren*{1+\frac{(1-\eta)x}{\eta}}} : 0 \le x \le t^2} \\
    =& \frac{1}{\gamma}\sup_{x \in [0,1]}\cbrace*{\frac{1}{1+\frac{(1-\eta)x}{\eta}}\abs*{\int_\gamma^{\gamma\paren*{1+\frac{(1-\eta)x}{\eta}}}\phi^{(1)}(s)ds - \frac{(1-\eta)x}{\eta}\phi(\gamma)} : 0 \le x \le t^2} \\
    =& \frac{1}{\gamma}\sup_{x \in [0,1]}\cbrace*{\frac{1}{1+\frac{(1-\eta)x}{\eta}}\abs*{\int_\gamma^{\gamma\paren*{1+\frac{(1-\eta)x}{\eta}}}\phi^{(1)}(s)ds} : 0 \le x \le t^2}.
  \end{align}
  From the Taylor theorem and the assumption that $\phi^{(1)}(0) = 0$, we have
  \begin{align}
    & \abs*{\int_\gamma^{\gamma\paren*{1+\frac{(1-\eta)x}{\eta}}}\phi^{(1)}(s)ds} \\
    \ge&  \abs*{\int_\gamma^{\gamma\paren*{1+\frac{(1-\eta)x}{\eta}}}\int_0^s\phi^{(2)}(s')ds'ds}
  \end{align}
  Letting $p_0 = 1\land(c'_2/W_2)^{1/(\alpha-2)}$, $\abs*{\phi^{(2)}(p)} \ge W_2p^{\alpha-2} - c'_2 \ge 0$ for $(0,p_0]$. From continuousness of $\phi^{(2)}$, $\phi^{(2)}(x)$ has the same sign for $x \in (0,p_0]$. For sufficiently small $\gamma$ such that $\gamma(1+\frac{(1-\eta)t^2}{\eta}) \le p_0$, we have
  \begin{align}
    &  \abs*{\int_\gamma^{\gamma\paren*{1+\frac{(1-\eta)x}{\eta}}}\int_0^s\phi^{(2)}(s')ds'ds} \\
    \ge& \abs*{\int_\gamma^{\gamma\paren*{1+\frac{(1-\eta)x}{\eta}}}\int_0^sW_2(s')^{\alpha-2} - c'_2 ds'ds} \\
    =& \abs*{\int_\gamma^{\gamma\paren*{1+\frac{(1-\eta)x}{\eta}}}\frac{W_2}{\alpha-1}s^{\alpha-1} - c'_2s ds} \\
    =& \abs*{\frac{W_2}{\alpha(\alpha-1)}\gamma^{\alpha}\paren*{\paren*{1+\frac{(1-\eta)x}{\eta}}^{\alpha}-1} - \frac{c'_2\gamma^2}{2}\paren*{\paren*{\frac{(1-\eta)x}{\eta}}^2-1}}.
  \end{align}
  Set $\eta = x/2$. Then, we have
  \begin{align}
    &  \abs*{\int_\gamma^{\gamma\paren*{1+\frac{(1-\eta)x}{\eta}}}\phi^{(2)}(\xi(s))ds} \\
    \ge& \abs*{\frac{W_2}{\alpha(\alpha-1)}\gamma^{\alpha}\paren*{\paren*{1+2\paren*{1-\frac{x}{2}}}^{\alpha}-1} - \frac{c'_2\gamma^2}{2}2\paren*{4\paren*{1-\frac{x}{2}}^2-1}} \\
    \ge& \gamma^{\alpha-1}\paren*{\frac{W_2}{\alpha(\alpha-1)}(2^{\alpha-1}-1)-3c'_2\gamma^{2-\alpha}} \\
    \ge& \Omega(\gamma^\alpha) \as \gamma \to 0.
  \end{align}
  Hence,
  \begin{align}
    \omega^1_\varphi(\phi^\star_{t^2/2,\gamma},t) \ge \Omega(\gamma^{\alpha-1}). \label{eq:lower-omega-1}
  \end{align}

  Let $\phi^{\star}_\gamma(x) = \phi^\star(\gamma(1+2L^2x^2))$. Then, we have $E_L(\phi^\star,[\gamma,2L^2\gamma]) = E_L(\phi^{\star}_\gamma,[-1,1])$. From the Jackson inequality, we have
  \begin{align}
    E_L(\phi^{\star}_\gamma,[-1,1]) \lesssim \omega_1(\phi^{\star}_\gamma,L^{-1})
  \end{align}
  For any $x,y \in (-1,1)$, we have
  \begin{align}
    &\abs*{\phi^{\star}_\gamma(x)-\phi^{\star}_\gamma(y)} \\
    \le& \abs*{\int_y^x 4\gamma L^2s\phi^{\star(1)}(\gamma(1+2L^2s^2))ds} \\
    \le&  \abs*{\int_y^x\frac{4L^2s}{1+2L^2s^2}\abs*{\phi^{(1)}(\gamma(1+2L^2s^2))}ds} + \abs*{\int_y^x\frac{4\gamma L^2s}{\gamma^2(1+2L^2s^2)^2}\abs*{\phi(\gamma(1+2L^2s^2))}ds} \\
    \le& \abs*{\int_y^x\frac{4L^2s}{1+2L^2s^2}\int_0^{\gamma(1+2L^2s^2)}\abs*{\phi^{(2)}(s')}ds'ds} + \abs*{\int_y^x\frac{4\gamma L^2s}{\gamma^2(1+2L^2s^2)^2}\int_\gamma^{\gamma(1+2L^2s^2)}\int_0^{s}\abs*{\phi^{(2)}(s')}ds''ds'ds}.
  \end{align}
  There exists $W'_2 \ge W_2$ such that $\abs*{\phi^{(2)}(p)} \le W'_2p^{\alpha-2}$ because $p^{\alpha-2} \ge 1$ for $p \in (0,1)$. For the first term, we have
  \begin{align}
    &\abs*{\int_y^x\frac{4L^2s}{1+2L^2s^2}\int_0^{\gamma(1+2L^2s^2)}\abs*{\phi^{(2)}(s')}ds'ds} \\
    \le& \abs*{\int_y^x\frac{4L^2s}{1+2L^2s^2}\int_0^{\gamma(1+2L^2s^2)}W'_2(s')^{\alpha-2}ds'ds}\\
    =& 2W'_2\gamma^{\alpha-1}\abs*{\int_y^x\frac{2L^2s(1+2L^2s^2)^{\alpha-2}}{\alpha-1}ds} \\
    \le& 2W'_2\gamma^{\alpha-1}\abs*{\int_y^x\frac{2^{\alpha-1}L^{2\alpha-2}s^{2\alpha-3}}{\alpha-1}ds} \\
    \le& \frac{2^\alpha W'_2}{(\alpha-1)(2\alpha-2)}L^{2\alpha-2}\abs*{x^{2\alpha-2}-y^{2\alpha-2}} \lesssim \gamma^{\alpha-1}L^{2\alpha-2}\abs*{x-y}^{2\alpha-2}.
  \end{align}
  For the second term, we have
  \begin{align}
    &\abs*{\int_y^x\frac{4\gamma L^2s}{\gamma^2(1+2L^2s^2)^2}\int_\gamma^{\gamma(1+2L^2s^2)}\int_0^{s}\abs*{\phi^{(2)}(s')}ds''ds'ds} \\
    \le& \abs*{\int_y^x\frac{4\gamma L^2s}{\gamma^2(1+2L^2s^2)^2}\int_\gamma^{\gamma(1+2L^2s^2)}\int_0^{s}W'_2(s'')^{\alpha-2}ds''ds'ds}\\
    \le& \frac{1}{\alpha(\alpha-1)}\abs*{\int_y^x\frac{4\gamma L^2s}{\gamma^2(1+2L^2s^2)^2}W'_2\gamma^\alpha\paren*{(1+2L^2s^2)^\alpha-1}ds} \\
    \le& \frac{2W'_2\gamma^{\alpha-1}}{\alpha(\alpha-1)}\abs*{\int_y^x2L^2s(1+2L^2s^2)^{\alpha-2}ds} \lesssim \gamma^{\alpha-1}L^{2\alpha-2}\abs*{x-y}^{2\alpha-2}.
  \end{align}
  Hence, $E_L(\phi^{\star}_\gamma,[-1,1]) \lesssim \omega_1(\phi^{\star}_\gamma,L^{-1}) \lesssim \gamma^{\alpha-1}$.

  With the definition of $\omega^1_\varphi(f,t)$, we have the direct result $E_L(f, [-1,1]) \lesssim \omega^1_\varphi(f,L^{-1})$ if $L \ge 1$. Also, we have the converse result $\frac{1}{L}\sum_{m=0}^LE_m(f, [-1,1]) \gtrsim \omega^1_\varphi(f,L^{-1})$~\autocite{ditzian2012moduli}. Let $L'$ be an integer such that $L' = c_\ell L$ where $c_\ell  > 1$. Then, we have
  \begin{align}
    & E_L(\phi^\star,[\gamma,2L^2\gamma])\\
     \ge& \frac{1}{L' - L}\sum_{m = L+1}^{L'} E_m(\phi^\star,[\gamma,2(m^2\lor1)\gamma]) \\
     \ge& \frac{1}{L'}\sum_{m = L+1}^{L'} E_m(\phi^\star,[\gamma,2(m^2\lor1)\gamma]) \\
     \ge& \frac{1}{L'}\sum_{m = 0}^{L'}E_m(\phi^\star,[\gamma,2(m^2\lor1)^2\gamma]) - \frac{1}{L'}E_0(\phi^\star,[\gamma,2\gamma]) - \frac{1}{L'}\sum_{m=1}^LE_m(\phi^\star,[\gamma,2m^2\gamma]). \label{eq:lower-best-approx}
  \end{align}
  For any $x \in (-1,1)$, we have
  \begin{align}
    \abs*{\phi^\star_{1/2,\gamma}(x)} =& \frac{1}{\gamma(3/2+x/2)}\abs*{\phi\paren*{\gamma(3/2+x/2)}} \\
    \le& \frac{1}{\gamma}\abs*{\int_\gamma^{\gamma(3/2+x/2)}\phi^{(1)}(s)ds} \\
    \le& \frac{1}{\gamma}\int_\gamma^{\gamma(3/2+x/2)}\int_0^s\abs*{\phi^{(2)}(s')}ds'ds \\
    \le& \frac{1}{\gamma}\int_\gamma^{\gamma(3/2+x/2)}\int_0^s\paren*{W_2s^{'\alpha-2}+c_2}ds'ds \\
    =& \frac{W_2}{\gamma(\alpha-1)\alpha}\gamma^{\alpha}\paren*{(3/2+x/2)-1}^\alpha + \frac{c_2}{2}\gamma^2\paren*{(3/2+x/2)-1}^2 \lesssim \gamma^{\alpha-1},
  \end{align}
  Thus, $E_0(\phi^\star,[\gamma,2\gamma]) \lesssim \gamma^{\alpha-1}$. Applying the converse result and the fact $E_L(\phi^{\star}_\gamma,[-1,1]) \lesssim \gamma^{\alpha-1}$ yields that there are constants $C > 0$, $C' > 0$, and $C'' > 0$ such that
  \begin{align}
    & E_L(\phi^\star,[\gamma,2L^2\gamma])\\
    \ge& C'\gamma^{\alpha-1} - \frac{C'\gamma^{\alpha-1}}{L'} - \frac{C''}{c_\ell} \\
    \ge& C'\gamma^{\alpha-1} - \frac{C'\gamma^{\alpha-1}}{c_\ell} - \frac{C''}{c_\ell}.
  \end{align}
  Thus, by taking sufficiently large $c_\ell$, there is $c > 0$ such that
  \begin{align}
    \limsup_{L \to \infty, \gamma \to 0 : \gamma \le 1/2L^2}\gamma^{1-\alpha}E_L\paren*{\phi_\gamma^\star, [\gamma,2L^2\gamma]} > c.
  \end{align}
\end{proof}

\section{Helper Lemmas}
\begin{lemma}\label{lem:bound-sum-alpha}
  Given $\alpha \in [0,1]$, $\sup_{P \in \dom{M}_k} \sum_{i=1}^k p_i^\alpha = k^{1-\alpha}$.
\end{lemma}
\begin{proof}[Proof of \cref{lem:bound-sum-alpha}]
  If $\alpha = 1$, the claim is obviously true. Thus, we assume $\alpha < 1$. We introduce the Lagrange multiplier $\lambda$ for a constraint $\sum_{i=1}^n p_i = 1$, and let the partial derivative of $\sum_{i=1}^k p_i^\alpha + \lambda(1 - \sum_{i=1}^k p_i)$ with respect to $p_i$ be zero. Then, we have
  \begin{align}
   \alpha p_i^{\alpha-1} - \lambda = 0. \label{eq:sum-alpha-diff}
  \end{align}
  Since $p^{\alpha-1}$ is a monotone function, the solution of \cref{eq:sum-alpha-diff} is given as $p_i = (\lambda/\alpha)^{1/(\alpha-1)}$, i.e., the values of $p_1,...,p_k$ are same. Thus, the function $\sum_{i=1}^k p_i^\alpha$ is maximized at $p_i = 1/k$ for $i=1,...,k$. Substituting $p_i=1/k$ into $\sum_{i=1}^k p_i^\alpha$ gives the claim.
\end{proof}

\end{document}